\documentclass[11pt]{article}

 \usepackage[utf8]{inputenc}
 \usepackage[T1]{fontenc}
 \usepackage[english]{babel}

\usepackage[top=20mm, bottom=20mm, left=20mm, right=20mm]{geometry}

\usepackage{graphicx}
\usepackage{float}

\usepackage{amsmath,amsthm,amssymb}
\usepackage{bbm}

\usepackage[math]{blindtext}

\usepackage{setspace}
\usepackage{natbib}
\usepackage{hyperref}

\newtheorem{theorem}{Theorem}[section]

\newtheorem{lemma}[theorem]{Lemma}
\newtheorem{proposition}[theorem]{Proposition}
\newtheorem{formula}{Formula}

\newcommand{\res}{\mathrm{Res}}
\newcommand{\ud}{\mathrm{d}}

\begin{document}
\newgeometry{margin=2cm,top=2cm,bottom=2cm}

\title{
        Some pricing tools for the Variance Gamma model}

\author{Jean-Philippe Aguilar 
        \thanks{BRED Banque Populaire, 18 Quai de la R\^ap\'ee, FR-75012 Paris, 
        Email: jean-philippe.aguilar@bred.fr}}
\date{June 02, 2020}

\maketitle
\thispagestyle{empty}

\begin{abstract}
We establish several closed pricing formula for various path-independent payoffs, under an exponential L\'evy model driven by the Variance Gamma process. These formulas take the form of quickly convergent series and are obtained via tools from Mellin transform theory as well as from multidimensional complex analysis. Particular focus is made on the symmetric process, but extension to the asymmetric process is also provided. Speed of convergence and comparison with numerical methods (Fourier transform, quadrature approximations, Monte Carlo simulations) are also discussed; notable feature is the accelerated convergence of the series for short term options, which constitutes an interesting improvement of numerical Fourier inversion techniques. 

\bigskip

\noindent {\bfseries Keywords:} L\'evy Process; Variance Gamma Process; Stochastic Volatility; Option Pricing.

\noindent {\bfseries AMS subject classifications (MSC 2020):} 60E07, 60E10, 60H35, 65C30, 65T50, 91G20, 91G30.

\noindent {\bfseries JEL Classifications:} C00, C02, G10, G13.
\end{abstract}

\newpage
\setcounter{page}{1}

\section{Introduction}\label{sec:intro}

L\'evy processes are well-known in financial modeling for their ability to reproduce realistic features of asset return distributions, such as skewness or leptokurtosis, that Gaussian models fail to describe (see \cite{Bertoin96} and \cite{Cont04} for a complete overview of L\'evy processes and their applications to quantitative finance). However, such processes having stationary increments, the implied volatility surfaces they generate remain constant over time (as functions of moneyness and time horizon); this is in contradiction with the widely observed fact that market volatility appears to be non constant, notably because large changes in asset prices tend to cluster together, resulting in alternating periods of high and low variance (see a comprehensive discussion in \cite{Cont07}). 

In order to model this complex volatility behavior, several approaches have been introduced. A natural idea is to describe volatility itself by a positive stochastic process: such models can be either time-continuous,  like in \cite{Heston93}, or discrete, like in \cite{Heston00}. Behavioral and agents-based models have also been proposed, for instance in \cite{Lux00}, and focus on providing an economic explanation of clustering and regime switching phenomena. Another fruitful approach is to introduce a time change in the L\'evy process describing asset prices - see a discussion in \cite{Barndorff-Nielsen02,Carr03} for instance: the stochastic process governing the evolution of time is called "stochastic clock" (following the terminology of \cite{Geman09}) or "business time", while the L\'evy process driving the asset's returns is said to evolve in "operational time". A nice feature of this category of models is their direct interpretation, the introduction of a random clock traducing the fact that trading activities are not uniform in time, but on the contrary display an alternation of peak and less busy periods; as a (drifted) Brownian motion whose time follows a Gamma process, the Variance Gamma process (see all details in \cite{Madan98}) belongs to this category.


For the purpose of derivatives pricing, the Variance Gamma process is typically implemented within the framework of exponential L\'evy models, such as described for instance in \cite{Schoutens03,Tankov10}, and we will therefore speak of the (exponential) Variance Gamma model; such a model has been successfully tested on real market data and be shown to perform better than Black-Scholes or Jump-Diffusion models in multiple situations (e.g., for European-style options on the HSI index in \cite{Lam02} or for currency options in \cite{Madan05}; see also \cite{Seneta04} for a detailed fitting procedure). Let us also mention that several extensions of the Variance Gamma model have been introduced and shown to provide a good fit to observed market prices, for example by subordinating bivariate or multivariate Brownian motions in \cite{Luciano06,Semeraro08}, with applications to basket options calibration on the DJIA index in \cite{Linders15}; other recent extensions include, for instance, the possibility of negative jumps in the linear drift rate of the price process in \cite{Ivanov18}. 

Pricing contingent claims in the Variance Gamma model, however, remains a complicated task, which can be a limitation to its popularity. Some exact formulas have been derived in particular cases (a Black-Scholes like formula for the European call involving products of Bessel and hypergeometric functions in \cite{Madan98}, or generalized hyperbolic functions in \cite{Ivanov18}) but, in practice, numerical methods are favored. Particularly popular are methods based on evaluation of Fourier integrals (see the classic references \cite{Carr99,Lewis01}), because the characteristic function of the Variance Gamma process is known exactly, and admits a relatively simple form; other recent techniques include Fourier-related (COS) transforms for European options \citep{Fang08}, Gauss-Laguerre quadrature for the probability densities \citep{Loregian12}, local basis frame projection method for both vanilla (European and American) and exotic options \citep{Kirkby15,Kirkby18}, or space-time discretization of the time process \citep{Cantarutti18}. Let us also mention approaches based on PIDE and their resolution with finite difference schemes as introduced by \cite{Cont05}, and, of course, Monte Carlo simulations and their variance reduction extensions (see the classic monograph \cite{Glasserman04} and references therein).

In this article, we would like to show that it is possible to obtain tractable closed pricing formulas for European or digital prices, but also for more exotic payoffs (gap, power, log options), thanks to a remarkable factorization property in the Mellin space. More precisely, in this representation, it is possible to express the price of a path-independent option as the product of the Mellin transforms of the Variance Gamma density and of the option's payoff. Let us remark that it is not the first time that the Mellin transform is used in quantitative finance: for instance, in the usual Black-Scholes framework, it has been employed to establish pricing formulas for single asset or basket European and American options (in the latter case, in terms of the optimal exercise price) by \cite{Panini04,Frontczak08}. The Mellin transform is also closely related to the bilateral Laplace transform, which has found many applications in option valuation; let us mention European option pricing in time-changed Brownian models (Variance Gamma, Normal inverse Gaussian of generalized hyperbolic models) in \cite{Xing17}, or Geometric Asian option pricing in affine stochastic volatility models in \cite{Hubalek17}.

Although based on the Mellin transform, our method differs from the approaches described above, because our purpose is not to compute numerically the Mellin integral, but to express it under the form of quickly convergent series of multidimensional residues; moreover, it turns out that the series terms resume to powers of the model's parameters, and are straightforward to compute. We will particularly focus on the symmetric model (i.e., when the price process is a Brownian motion without drift), but will also show how to extend the technique to a drifted process. The technology, based on the Mellin factorization property and on residue calculation, has been previously implemented in exponential L\'evy models driven by spectrally negative processes (Black-Scholes or Finite Moment Log Stable models) in \cite{AK19,Aguilar19}; in the present article, we will demonstrate that the technique is also well suited to two-sided processes. In particular, we will see that the obtained series converge extremely fast for very short maturities, which constitutes an advantage over numerical tools (as noticed in \cite{Carr99}, very short maturities create highly oscillatory integrands which considerably slow down the numerical Fourier inversion process).

The paper is organized as follows: in section~\ref{sec:model}, we start by recalling fundamental concepts on the Variance Gamma process, and on option pricing in exponential L\'evy models; then, in section~\ref{sec:symmetricVG}, we establish a generic pricing formula in the case of a symmetric process, which allows to determine closed-form series formulas for the price of several options. In section~\ref{sec:asymmetricVG}, we extend this pricing formula to the case of an asymmetric process, and apply it to the case of a digital option. In section \ref{sec:tests}, we compare our results with several numerical techniques, such as Fourier inversion, Gauss-Laguerre approximation for the probability densities or Monte Carlo simulations (in both the symmetric and asymmetric cases). Last, section~\ref{sec:conclusion} is dedicated to concluding remarks and discussion of future work. For the reader's convenience, the paper is also equipped with two appendices: appendix \ref{app:Mellin} offers a review of the theory of the Mellin transform in $\mathbb{C}$ and $\mathbb{C}^n$, and appendix \ref{app:LG} provides more details on the Gauss-Laguerre quadrature and how it can be used to approximate option prices.

\section{Model definition}\label{sec:model}

We start by recalling some basic facts about the Variance-Gamma process, and how it is used for option pricing within the framework of exponential L\'evy models.

\subsection{The Variance Gamma process}\label{subsec:variance_gamma}

The Variance Gamma (VG) process is a pure jump L\'evy process with finite variation and infinite activity; it can be defined either as a time-changed Brownian motion, a difference of two Gamma processes, or a particular case of a tempered stable L\'evy process. In this subsection we briefly recall, without proofs, the main features of the VG process as seen from these different point of views; detailed proofs and calculations can be found e.g. in the original paper by Madan, Carr and Chang \cite{Madan98}.

\paragraph{Gamma subordination} 
In this approach, the VG process is obtained by evaluating a drifted Brownian motion at a random time following a Gamma process (see \cite{Bertoin99,Geman01} for general results on Gamma subordination and time changes in L\'evy processes). More precisely, let $(\Omega, \mathcal{F}, \{ \mathcal{F}_t \}_{t \geq 0}, \mathbb{P} )$ be a filtered probability space, let $\gamma(t,\mu,\nu)$ be a Gamma process (i.e. a process whose increments $\gamma(t+h,\mu,\nu) - \gamma(t,\mu,\nu)$ follow a Gamma distribution with mean $\mu h $ and variance $\nu h$), and let $B(t,\theta,\sigma)$ be a drifted Brownian motion:
\begin{equation}
    B(t,\theta,\sigma) \, : = \, \theta t \, + \, \sigma W(t) \,\, , \,\,\,\,\  W(t) \sim \mathcal{N}(0,t)
    .
\end{equation}
Then, the VG process is the stochastic process defined by:
\begin{equation}\label{Variance-Gamma}
    X(t,\sigma,\nu,\theta) \, := \, B( \gamma(t,1,\nu) , \theta ,\sigma)
\end{equation}
where $\sigma$ is the \textit{scale parameter}, $\nu$ the \textit{kurtosis parameter} and $\theta$ is the \textit{asymmetry parameter}; the density function of the process is obtained by integrating the normal density conditionally to the Gamma time-change:
\begin{equation}\label{VG_density_1}
    f_{\sigma,\nu,\theta}(x,t) \, = \, \int\limits_0^\infty \, \frac{1}{\sigma\sqrt{2\pi y}} e^{-\frac{(x-\theta y)^2}{2\sigma^2 y}} \, \frac{y^{\frac{t}{\nu}-1}}{\nu^{\frac{t}{\nu}}\Gamma(\frac{t}{\nu})} \, e^{-\frac{y}{\nu}} \, \ud y
    .
\end{equation}
The density \eqref{VG_density_1} can be expressed in terms of special functions as:
\begin{equation}\label{VG_density_2}
    f_{\sigma,\nu,\theta}(x,t) \, = \, \frac{2 e^{\theta\frac{x}{\sigma^2}}}{\nu^{\frac{t}{\nu}} \sqrt{2\pi} \sigma \Gamma(\frac{t}{\nu})} \, \left( \frac{x^2}{\frac{2\sigma^2}{\nu} + \theta^2}  \right)^{\frac{t}{2\nu}-\frac{1}{4}} \, K_{\frac{t}{\nu}-\frac{1}{2}} \left( \frac{1}{\sigma^2} \sqrt{\frac{2\sigma^2}{\nu} + \theta^2} |x| \right)
\end{equation}
where $K_\alpha(X)$ denotes the modified Bessel function of the second kind (see definition and properties in \cite{Abramowitz72}). The characteristic function admits a simple representation in the Fourier space:
\begin{equation}\label{characteristic_function}
    \Phi_{\sigma,\nu,\theta}(u,t) \,  := \, \mathbb{E}^{\mathbb{P}}  \, \left[ e^{i u X(t,\sigma,\nu,\theta) } \right] \, = \,  \left(  \frac{1}{1- i \theta\nu u + \frac{\sigma^2\nu}{2}u^2} \right)^{\frac{t}{\nu}}
\end{equation}
and, if we define the L\'evy symbol (or characteristic exponent) by:
\begin{equation}\label{characteristic_exponent}
     \Psi_{\sigma,\nu,\theta}(u) \,  := \, -\log \mathbb{E}^{\mathbb{P}} \, \left[ e^{i u X(1,\sigma,\nu,\theta) } \right] \,  = \, \frac{1}{\nu} \, \log \left(  1- i \theta\nu u + \frac{\sigma^2\nu}{2}u^2 \right)
     ,
\end{equation}
then the following holds true:
\begin{equation}
    \Phi_{\sigma,\nu,\theta} (u,t) \, = \,  e^{-t\Psi_{\sigma,\nu,\theta}(u)} 
    .
\end{equation}

\paragraph{Decomposition} 
As a consequence of the finite variations property, the VG process can be written down as a difference of two Gamma processes:
\begin{equation}
    X(t,\sigma,\nu,\theta) \, = \, \gamma^+ (t,\mu_+,\nu_+) \, - \, \gamma^- (t,\mu_-,\nu_-)
\end{equation}
where the parameters are defined by
\begin{equation}\label{gamma_parameters}
    \left\{
    \begin{aligned}
        & \mu_\pm \, := \, \frac{1}{2} \sqrt{\theta^2 + \frac{2\sigma^2}{\nu}} \, \pm \, \frac{\theta}{
        2} \\
        & \nu_\pm \, := \, \mu_\pm^2 \nu
        .
    \end{aligned}
    \right.
\end{equation}

\paragraph{L\'evy measure}
The VG process is a L\'evy process without Brownian component nor additional drift, i.e., the L\'evy-Khintchine representation for the characteristic exponent \eqref{characteristic_exponent} is:
\begin{equation}
    \Psi_{\sigma,\nu,\theta}(k) \, = \, \int_{\mathbb{R}\backslash \{0\} } \, ( 1- e^{i k x}  ) \, \Pi_{\sigma,\nu,\theta} (\ud x)
\end{equation}
where the L\'evy measure can be expressed as:
\begin{equation}\label{Levy_measure_1}
    \Pi_{\sigma,\nu,\theta}( \ud x) \, = \, \frac{1}{\nu} \left( \frac{e^{-\frac{1}{\mu_-\nu}|x|}}{|x|} \mathbbm{1}_{ \{ x<0 \} } \, + \, \frac{e^{-\frac{1}{\mu_+\nu}x}}{x} \mathbbm{1}_{ \{ x>0 \} }  \right)  \ud x
    .
\end{equation}
The representation \eqref{Levy_measure_1} shows that the VG process is actually a particular case of a CGMY process of \cite{Carr02}, with \textit{overall activity level} $C=\frac{1}{\nu}$, \textit{skewness parameters} $G=\frac{1}{\mu_-\nu}$ and $M=\frac{1}{\mu_+\nu}$, and \textit{fine structure parameter} $Y=1$. Note that, using the transformation \eqref{gamma_parameters}, the L\'evy measure \eqref{Levy_measure_1} can be re-written as:
\begin{equation}\label{Levy_measure_2}
    \Pi_{\sigma,\nu,\theta}( \ud x) \, = \, \frac{e^{\frac{\theta x}{\sigma^2}}}{\nu |x|} e^{ -\frac{\sqrt{  \frac{\theta^2}{\sigma^2} + \frac{2}{\nu} }}{\sigma} |x| } \, \ud x
\end{equation}
and is symmetric around the origin when $\theta = 0$.

\subsection{Exponential VG model and option pricing}\label{subsec:exponential_process}

Let us introduce the exponential VG model, following the classical setup of exponential L\'evy models such as defined e.g. in \cite{Schoutens03,Tankov10}.
\paragraph{Model specification}
Let $T>0$, and let $S_t$ denote the value of a financial asset at time $t\in[0,T]$; we assume that it can be modeled as the realization of a stochastic process $\{ S_t \}_{t \in [0,T]}$ on the canonical space $\Omega = \mathbb{R}_+$ equipped with its natural filtration, and that, under the \textit{risk-neutral measure} $\mathbb{Q}$, its instantaneous variations can be written down in local form as:
\begin{equation}\label{SDE_exp}
    \frac{\ud S_t}{S_t} \, = \, (r - q) \, \ud t \, + \,  \, \ud X (t,\sigma,\nu,\theta)
    .
\end{equation}
In the stochastic differential equation \eqref{SDE_exp}, $r\in\mathbb{R}$ is the risk-free interest rate and $q\in\mathbb{R}$ is the dividend yield (both assumed to be deterministic and continuously compounded), and $\{ X (t,\sigma,\nu,\theta) \}_{t \in [0,T]}$ is the VG process; for the simplicity of notations, we will assume that $q = 0$, but all the results of the paper remain valid when replacing $r$ by $r-q$.

The solution to \eqref{SDE_exp} is the \textit{exponential process} defined by:
\begin{equation}\label{Solution_exp}
    S_T \, = \, S_t e^{(r+\omega_{\sigma,\nu,\theta})\tau + X (\tau,\sigma,\nu,\theta) }
\end{equation}
where $\tau:=T-t$ is the horizon (or time-to-maturity), and $\omega_{\sigma,\nu,\theta}$ is a \textit{martingale adjustment} (also called convexity adjustment, or compensator) computed in a way that the discounted stock price is a $\mathbb{Q}$-martingale, i.e. that the relation $\mathbb{E}_t^{\mathbb{Q}} [ S_T ] = e^{r\tau} S_t$ is satisfied. From \eqref{Solution_exp}, this resumes to the condition:
\begin{equation}
    \mathbb{E}_t^{\mathbb{Q}} \left[ e^{\omega_{\sigma,\nu,\theta}\tau + X (\tau,\sigma,\nu,\theta)  } \right] \, = \, 1
    ,
\end{equation}
or, equivalently, in terms of the L\'evy symbol \eqref{characteristic_exponent}:
\begin{equation}\label{omega_def}
    \omega_{\sigma,\nu,\theta} \, = \, \Psi_{\sigma,\nu,\theta}(-i) \, = \, \frac{1}{\nu}\log \left( 1 - \theta\nu - \frac{\sigma^2\nu}{2}  \right)
    .
\end{equation}
It is interesting to note that, if $\theta=0$, then \eqref{omega_def} coincides with the Gaussian martingale adjustment in the low variance regime:
\begin{equation}
    \omega_{\sigma,\nu,0} \, \underset{\nu\rightarrow 0}{\longrightarrow} \, -\frac{\sigma^2}{2}
\end{equation}
and, in this limit, the exponential VG model degenerates into the usual Black-Sholes model.

\paragraph{Option pricing}
Let $N\in\mathbb{N}$ and $\mathcal{P}: \mathbb{R}_+^{1+N} \rightarrow \mathbb{R}$ be a time independent payoff function, i.e. a function depending only on the terminal price $S_T$ and on some positive parameters $K_n$, $n= 1\dots N$:
\begin{equation}
    \mathcal{P} \, : (S_T,K_1, \dots , K_N) \, \rightarrow \, \mathcal{P} (S_T,K_1, \dots , K_N) \, := \, \mathcal{P} (S_T, \underline{K} )
    .
\end{equation}
The value at time $t$ of a contingent claim $\mathcal{C}_{\sigma,\nu,\theta}$ with payoff $\mathcal{P}(S_T,\underline{K})$ is equal to the risk-neutral conditional expectation of the discounted payoff:
\begin{equation}\label{Risk-neutral_1}
   \mathcal{C}(S_t,\underline{K},r,t,T,\sigma,\nu,\theta) \, = \, \mathbb{E}_t^{\mathbb{Q}} \left[e^{-r(T-t)} \mathcal{P} (S_T,\underline{K}) \right]
   .
\end{equation}
As the VG process admits the density $f_{\sigma,\nu,\theta}(x,t)$, then, using \eqref{Solution_exp}, we can re-write \eqref{Risk-neutral_1} by integrating all possible realizations for the terminal payoff over the VG density :
\begin{equation}\label{Risk-neutral_2}
    \mathcal{C}(S_t,\underline{K},r,\tau,\sigma,\nu,\theta) \, = \, e^{-r\tau} \, \int\limits_{-\infty}^{+\infty} \, \mathcal{P} \left( S_t e^{(r+\omega_{\sigma,\nu,\theta})\tau +x} , \underline{K}  \right) f_{\sigma,\nu,\theta}(x,\tau) \, \ud x
    .
\end{equation}
In all the following and to simplify the notations, we will forget the $t$ dependence in the stock price $S_t$. 

\section{Symmetric Variance Gamma process}\label{sec:symmetricVG}

One speaks of a \textit{symmetric Variance Gamma process} when $\theta=0$, that is, when the process corresponds to a time-changed Brownian motion without drift; it was first considered in \cite{Madan90}. In this section we start by showing that, under symmetric VG process, the price of a contingent claim admits a factorized representation in the Mellin space; then, we apply this pricing formula to the valuation of various payoffs.

\subsection{Pricing formula}\label{subsec:symmetric_pricing_formula}

Let $\sigma,\nu > 0$ and let us denote the density of the symmetric VG process by $f_{\sigma,\nu}(x,t):=f_{\sigma,\nu,0}(x,t)$. 
\begin{lemma}\label{lemma:density_symmetricVG}
    Let $F_{\nu}^*$ be the meromorphic function on $\mathbb{C}$ defined by
    \begin{equation}\label{F_nu}
        F_{\nu}^*(s_1) \, = \,  \Gamma \left( \frac{s_1-\frac{\tau}{\nu} + \frac{1}{2}}{2} \right) \Gamma \left( \frac{s_1 + \frac{\tau}{\nu} - \frac{1}{2}}{2} \right) 
        .
    \end{equation}
    Then, for any $ c_1 > | \frac{\tau}{\nu} - \frac{1}{2} | $, the following Mellin-Barnes representation holds true:
    \begin{equation}\label{Density_symmetricVG_MB}
        f_{\sigma,\nu}(x,\tau) \, = \, \frac{1}{ 2 \sqrt{\pi} \Gamma(\frac{\tau}{\nu}) } \, \int\limits_{c_1 - i\infty}^{c_1 + i\infty} \, F_{\nu}^*(s_1) \, (\sigma\sqrt{2\nu})^{s_1-\frac{\tau}{\nu} - \frac{1}{2}} |x|^{-s_1+\frac{\tau}{\nu} - \frac{1}{2}} \, \frac{\ud s_1}{2i\pi}
        .
    \end{equation}
\end{lemma}
\begin{proof}
    Taking $\theta = 0$ in \eqref{VG_density_2} yields:
    \begin{equation}
        f_{\sigma,\nu}(x,\tau) \, = \, \frac{2}{ \sqrt{2\pi} \Gamma(\frac{\tau}{\nu}) (\sigma^2\nu)^{\frac{\tau}{\nu}} } \, 
        \left( \frac{ |x| }{ \frac{1}{\sigma} \sqrt{\frac{2}{\nu}} }  \right)^{\frac{\tau}{\nu}-\frac{1}{2}} \, 
        K_{\frac{\tau}{\nu}-\frac{1}{2}} \left( \frac{1}{\sigma} \sqrt{\frac{2}{\nu}} |x| \right)
    \end{equation}
    which, as expected, is symmetric around 0. The Mellin-Barnes representation \eqref{Density_symmetricVG_MB} is easily obtained from the Mellin transform (see table \ref{tab:Mellin} in appendix \ref{app:Mellin}, or \cite{Bateman54} p. 331):
    \begin{equation}
        K_{\alpha} (X) \, \longrightarrow \, 2^{s_1-2}  \, \Gamma \left( \frac{s_1-\alpha}{2} \right) \Gamma \left( \frac{s_1+\alpha}{2} \right) \, \, , \,\,\,\,\, X>0
\end{equation}
    which converges for $Re(s_1) > |Re (\alpha)|$, and by applying the Mellin inversion formula (see appendix \ref{app:Mellin}, or \cite{Flajolet95}).
\end{proof}

Let us now introduce the double-sided Mellin transform of the payoff function:
\begin{equation}\label{Payoff_transform}
    P_{\sigma,\nu}^*(s_1) \, = \, \int\limits_{-\infty}^{\infty} \, \mathcal{P} \left( S e^{(r+\omega_{\sigma,\nu})\tau +x} , \underline{K}  \right) \, |x|^{-s_1 + \frac{\tau}{\nu} - \frac{1}{2}} \, \ud x
\end{equation}
where the symmetric martingale adjustment resumes to:
\begin{equation}\label{symmetric_omega}
    \omega_{\sigma,\nu} \, := \, \omega_{\sigma,\nu,0} \, = \, \frac{1}{\nu}\log \left( 1 - \frac{\sigma^2\nu}{2}  \right)
    ,
\end{equation}
and let us assume that the integral \eqref{Payoff_transform} converges on some interval $s_1 \in (c_-,c_+)$, $c_-<c_+$. Then, as a direct consequence of the pricing formula \eqref{Risk-neutral_2} and of lemma~\ref{lemma:density_symmetricVG}, we have:

\begin{proposition}[Factorization in the Mellin space]
    \label{prop:symmetric_factorization}
    Let $c_1\in (\tilde{c}_-,\tilde{c}_+)$ where $
    (\tilde{c}_-,\tilde{c}_+) := (c_-,c_+) \cap (| \frac{\tau}{\nu} - \frac{1}{2} | , \infty)$ is assumed to be nonempty. Then the value at time $t$ of a contingent claim $\mathcal{C}$ with maturity $T$ and payoff $\mathcal{P}(S_T,\underline{K})$ is equal to:
    \begin{equation}\label{Risk-neutral:convolution}
        \mathcal{C} (S,\underline{K},r,\tau,\sigma,\nu) \, = \, \frac{e^{-r\tau}}{ 2 \sqrt{\pi} \Gamma(\frac{\tau}{\nu}) } \int\limits_{c_1-i\infty}^{c_1+i\infty} \, P_{\sigma,\nu}^*(s_1) F_{\nu}^*(s_1) \, \left( \sigma\sqrt{2\nu} \right)^{ s_1-\frac{\tau}{\nu} - \frac{1}{2} } \, \frac{\ud s_1}{2i\pi}
        .
    \end{equation}
\end{proposition}
The factorized form \eqref{Risk-neutral:convolution} turns out to be a very practical tool for option pricing. Indeed, as an integral along a vertical line in the complex plane, it can be conveniently expressed as a sum of residues associated to the singularities of the integrand, i.e., schematically: 
\begin{equation}
    \frac{e^{-r\tau}}{ 2 \sqrt{\pi} \Gamma(\frac{\tau}{\nu}) } \times \sum \, \left[ \textrm{residues of }P_{\sigma,\nu}^*(s_1) F_{\nu}^*(s_1) \times \textrm{powers of } \sigma\sqrt{2\nu} \right]
    .
\end{equation}
However, depending on the payoff's complexity, it can be necessary to introduce a second Mellin variable in order to evaluate $P_{\sigma,\nu}^*(s_1)$. In this case, it is possible to express the arising multiple complex integral as a sum of multidimensional residues, in virtue of a multidimensional generalization of the Jordan lemma which goes as follow.  Let the $\underline{a}_k$, $\underline{\tilde{a}}_j$ be vectors in $\mathbb{C}^n$, let the $b_k$, $\tilde{b}_j$ be complex numbers and let "." represent the euclidean scalar product. Assume that we are interested in evaluating an integral of the type
\begin{equation}\label{MB_Cn}
    \int\limits_{\underline{c}+i\mathbb{R}^n} \, \omega
\end{equation}
where the vector $\underline{c}$ belongs to the region of convergence $\mathrm{P}\subset\mathbb{C}^n$ of the integral \eqref{MB_Cn}, and where $\omega$ is a complex differential n-form reading
\begin{equation}
    \omega \, = \, \frac{\Gamma(\underline{a}_1.\underline{s} + b_1) \dots \Gamma(\underline{a}_m.\underline{s} + b_m)}{\Gamma(\underline{\tilde{a}}_1.\underline{s} + \tilde{b}_1) \dots \Gamma(\underline{\tilde{a}}_l.\underline{s} + \tilde{b}_l)} \, x_1^{-s_1} \, \dots \, x_n^{-s_n} \, \frac{\ud s_1 \dots \ud s_n}{(2i\pi)^n}  \hspace*{1cm} \, x_1, \dots x_n \in\mathbb{R}
    .
\end{equation}
The singular sets $D_k$ induced by the singularities of the Gamma functions
\begin{equation}
    D_k \, := \, \{ \underline{s}\in\mathbb{C}^n \, , \, \underline{a}_k.\underline{} + b_k = -n_k \, , \, n_k \in\mathbb{N}   \} \hspace*{1cm}  k=1 \dots m
\end{equation}
are called the \textit{divisors} of $\omega$. The \textit{characteristic vector} of $\omega$ is defined to be
\begin{equation}\label{Delta_n}
    \Delta \, = \, \sum\limits_{k=1}^m \underline{a}_k \, - \, \sum\limits_{j=1}^l \underline{\tilde{a}}_j
\end{equation}
and, for $\Delta\neq 0$, its \textit{admissible region} is:
\begin{equation}
    \Pi_\Delta \, := \, \{ \underline{s}\in\mathbb{C}^n \, , \, Re( \Delta . \underline{s} ) \, < \, Re( \Delta . \underline{c})  \}
    .
\end{equation}
As a consequence of the Stirling approximation at infinity for the Gamma function (see e.g. \cite{Abramowitz72}), $\omega$ decays exponentially fast in $\Pi_{\Delta}$ and 
\begin{equation}\label{Tsikh_residue}
    \int\limits_{\underline{c}+i\mathbb{R}^n} \, \omega \, = \, \sum\limits_{\Pi} \mathrm{Res} \, \omega
\end{equation}
for any \textit{cone} $\Pi\subset\Pi_{\Delta}$ whose faces are interested by exactly one of the divisors of $\omega$; such a cone is said to be \textit{compatible} with $\omega$, and generalizes the notion of left or right half plane in $\mathbb{C}$. If $\Delta=0$, then \eqref{Tsikh_residue} holds in every compatible cone. Readers interested in this very rich theory will find full details in \cite{Passare97} (see also \cite{Zhdanov98} for more details in the case $n=2$); see also appendix \ref{app:Mellin}.

\subsection{European and digital prices}\label{subsec:symmetric_payoffs}

In all the following, the \textit{forward strike price} $F$ and the \textit{log-forward moneyness} $k$ are defined to be:
\begin{equation}\label{moneyness}
   F \, := \, Ke^{-r\tau} \hspace*{1cm}
   k \, := \, \log\frac{S}{F} \, = \, \log\frac{S}{K} + r\tau
   .
\end{equation}
The \textit{re-scaled volatility} $\sigma_\nu$ and the \textit{(risk-neutral) moneyness} $k_\nu$ are:
\begin{equation}\label{asymmetricVG_notations}
    \sigma_\nu \, := \, \sigma\sqrt{\frac{\nu}{2}} \hspace*{1cm}
    k_{\sigma,\nu} \, := \, k + \omega_{\sigma,\nu}\tau
    .
\end{equation}
Instead otherwise stated, $k_{\sigma,\nu}$ will be our standard measure of moneyness, that is:
\begin{equation}
\left\{
    \begin{aligned}
    k_{\sigma,\nu} \, < \, 0 \, : & \hspace*{0.5cm} \textrm{Out-of-the-money (OTM) call} \\
    k_{\sigma,\nu} \, = \, 0 \, : & \hspace*{0.5cm} \textrm{At-the-money (ATM) call} \\
    k_{\sigma,\nu} \, > \, 0 \, : & \hspace*{0.5cm} \textrm{In-the-money (ITM) call}
    .
    \end{aligned}
\right.
\end{equation}
By an abuse of notations, we will denote the price of contingent claims either by $\mathcal{C}(S,K,r,\tau,\sigma,\nu)$ or $\mathcal{C}(k_{\sigma,\nu},\sigma_\nu)$. We will also denote  $\alpha:=\frac{\tau}{\nu}-\frac{1}{2}$ and assume $\alpha\notin \mathbb{Z}$ to avoid the degenerate case for for the Bessel function. Last, we will use the standard notation $[X]^+:= X \mathbbm{1}_{ \{X>0 \}}$.

\paragraph{Digital option (asset-or-nothing)} The asset-or-nothing call option is a simple exotic option consisting in receiving a unit of the underlying asset $S_T$, on the condition that $S_T$ is greater than a predetermined strike price $K$; the payoff is therefore:
\begin{equation}
    \mathcal{P}_{a/n} (S_T,K) \, := \, S_T \, \mathbbm{1}_{ \{ S_T  > K \}  }
    .
\end{equation}
\begin{formula}[Asset-or-nothing call] 
    \label{formula:symmetric_A/N}
    The value at time $t$ of an asset-or-nothing call option is: \\
    (i) (OTM price) If $ k_{\sigma,\nu} < 0$, 
        \begin{multline}\label{symmetric_A/N_OTM}
            C_{a/n}^-(k_{\sigma,\nu},\sigma_\nu) \, = \, \frac{F}{2\Gamma(\frac{\tau}{\nu})} \sum\limits_{\substack{n_1 = 0 \\ n_2 = 0}}^{\infty} \frac{(-1)^{n_1}}{n_1!}
            \left[
            \frac{\Gamma(\frac{-n_1+n_2+1}{2} + \alpha)}{\Gamma(\frac{-n_1+n_2}{2}+1)} 
            \left( \frac{-k_{\sigma,\nu}}{\sigma_\nu} \right)^{n_1} \sigma_\nu^{n_2} \right.
            \\
            \left.
            +
            \, 2  \, \frac{\Gamma(-2n_1-n_2-1 -2\alpha)}{\Gamma(-n_1+\frac{1}{2}-\alpha)} 
            \left( \frac{-k_{\sigma,\nu}}{\sigma_\nu} \right)^{2n_1 +1 +2\alpha} (-k_{\sigma,\nu})^{n_2}
            \right]
            .
        \end{multline}
    (ii) (ITM price) If $ k_{\sigma,\nu} > 0$,
        \begin{equation}\label{symmetric_A/N_ITM}
            C_{a/n}^+(k_{\sigma,\nu},\sigma_\nu) \, = \, S - C_{a/n}^-(k_{\sigma,\nu},-\sigma_\nu)
            .
        \end{equation}
        (iii) (ATM price) If $ k_{\sigma,\nu} = 0 $,
        \begin{equation}\label{symmetric_A/N_0}
            C_{a/n}^-(k_{\sigma,\nu},\sigma_\nu) \, = \,  C_{a/n}^+(k_{\sigma,\nu},\sigma_\nu) \, = \, \frac{F}{2\Gamma(\frac{\tau}{\nu})} \sum\limits_{n=0}^\infty \frac{\Gamma(\frac{n+1}{2}+\alpha)}{\Gamma(\frac{n}{2}+1)} \, \sigma_\nu^n 
            .
        \end{equation}        
\end{formula}
\begin{proof}
    To prove (i), we first note that, using notations \eqref{moneyness}, we have:
    \begin{equation}
        \mathcal{P}_{a/n} (Se^{(r+\omega_{\sigma,\nu})\tau+x},K) \, = \, Ke^{k_{\sigma,\nu}+x} \mathbbm{1}_{ \{x > -k_{\sigma,\nu}\} }
        .
    \end{equation}
    Using a Mellin-Barnes representation for the exponential term (see table \ref{tab:Mellin} in appendix \ref{app:Mellin}, or \cite{Bateman54} p. 312):
    \begin{equation}\label{MB_payoff_a/n}
        e^{k_{\sigma,\nu}+x} \, = \, \int\limits_{c_2 - i\infty}^{c_2 + i\infty} (-1)^{-s_
        2} \Gamma(s_2) (k_{\sigma,\nu} + x)^{-s_2} \, \frac{\ud s_2}{2i\pi} \hspace{1cm} (c_2 > 0)
    \end{equation}
    then the $P^*_{\sigma,\nu}(s_1)$ function \eqref{Payoff_transform} becomes:
    \begin{align}
        P^*_{\sigma,\nu}(s_1) & = \, K \int\limits_{c_2 - i\infty}^{c_2 + i\infty} (-1)^{-s_2} \Gamma(s_2) \int\limits_{-k_{\sigma,\nu}}^\infty (k_{\sigma,\nu} + x)^{-s_2} x^{-s_1+\alpha} \, \ud x \, \frac{\ud s_2}{2i\pi} \\
        & = K \int\limits_{c_2 - i\infty}^{c_2 + i\infty} (-1)^{-s_2} \frac{\Gamma(s_2)\Gamma(1-s_2)\Gamma(s_1+s_2-1-\alpha)}{\Gamma(s_1-\alpha)} (-k_{\sigma,\nu})^{-s_1-s_2+1+\alpha}   
        \, \frac{\ud s_2}{2i\pi}
    \end{align}
    where the $x$-integral is a particular case of a B\^eta integral, and converges because $-k_{\sigma,\nu}>0$. Inserting in the pricing formula \eqref{Risk-neutral:convolution} and using the Legendre duplication formula, the asset-or-nothing call thus reads:
    \begin{multline}\label{symmetric_A/N_OTM_integral}
        C_{a/n}^-(k_{\sigma,\nu},-\sigma_\nu) \, = \, \frac{F}{2\Gamma(\frac{\tau}{\nu})} \, \int\limits_{c_1-i\infty}^{c_1+i\infty} \int\limits_{c_2-i\infty}^{c_2+i\infty} \, (-1)^{-s_2} \, \times
        \\
         \frac{\Gamma(\frac{s_1+\alpha}{2})\Gamma(s_2)\Gamma(1-s_2)\Gamma(s_1+s_2-1-\alpha)}{\Gamma(s_1-\alpha)} (-k_{\sigma,\nu})^{-s_1-s_2+1+\alpha}  \sigma_\nu^{s_1-1-\alpha} \frac{\ud s_1 \ud s_2}{(2i\pi)^2 }
    \end{multline}
    where $(c_1,c_2)$ belongs to the convergence polyhedron of the integral $\mathrm{P}:=\{ (s_1,s_2)\in\mathbb{C}^2, 0 < Re(s_2 < 1, Re(s_1)>-\alpha,  Re(s_1+s_2) > 1+\alpha \}$. The characteristic vector \eqref{Delta_n} associated to \eqref{symmetric_A/N_OTM_integral} is equal to:
    \begin{equation}
        \Delta \, = \, 
        \begin{bmatrix}
            \frac{1}{2} \\ 1
        \end{bmatrix}
    \end{equation}
    resulting in the admissible region $\Pi_\Delta := \{ (s_1,s_2)\in\mathbb{C}^2, Re(s_2) < \frac{1}{2}(c_1 - Re(s_1)) +c_2 \}$. The cone $\Pi := \{ (s_1,s_2)\in\mathbb{C}^2 , Re(s_2)< 0 , Re (s_1+s_2) < 1 + \alpha \}$ satisfies $\Pi\subset\Pi_\Delta$ and, by definition, each divisor of \eqref{symmetric_A/N_OTM_integral} intersects exactly one of its faces; as a consequence of the residue theorem \eqref{Tsikh_residue}, the integral \eqref{symmetric_A/N_OTM_integral} equals the sum of the residues located in $\Pi$; these residues are associated with two series of poles of the integrand:
    \begin{itemize}
        \item[-] Series 1: when $\Gamma(\frac{s_1+\alpha}{2})$ and $\Gamma(s_2)$ are singular;
        \item[-] Series 2: when $\Gamma(s_1+s_2-1-\alpha)$ and $\Gamma(s_2)$ are singular.
    \end{itemize}
    The Gamma function being singular when its argument is a negative integer, poles of series 1 are located at $(s_1,s_2)=(-\alpha-2n_1,-n_2)$, $n_1,n_2\in\mathbb{N}$; a sequential application of the Cauchy formula yields the associated residues:
    \begin{multline}
        \frac{F}{2\Gamma(\frac{\tau}{\nu})} (-1)^{n_2}  \frac{2 (-1)^{n_1+n_2}}{n_1!n_2!}\frac{\Gamma(1+n_2)\Gamma(-2n_1-n_2-1-2\alpha)}{\Gamma(-n_1+\frac{1}{2}-\alpha)}(-k_{\sigma,\nu})^{2n_1+n_2+1+2\alpha}\sigma_\nu^{-2n_1-1-2\alpha}
    \end{multline}
    which, after simplification, is the second term in the r.h.s. of \eqref{symmetric_A/N_OTM}. Residues associated to the pole of series 2 can also be computed via the Cauchy formula (after making the change of variables $\tilde{s_1}:=s_1+s_2-1-\alpha$, $\tilde{s_2}:=s_2$), and give birth to the first term in the r.h.s. of \eqref{symmetric_A/N_OTM}, which completes the proof of (i).
    
    \noindent To prove (ii), it suffices to remark that, as $\{S_t\}_{t\in[0,T]}$ is a $\mathbb{Q}$-martingale, one has:
    \begin{equation}
        \mathbb{E}_t^\mathbb{Q} [ S_T \mathbbm{1}_{ \{ S_T > K \} } ] \, = \, Se^{r\tau} \, - \, \mathbb{E}_t^\mathbb{Q} [ S_T \mathbbm{1}_{ \{ S_T < K \} } ]
        .
    \end{equation}
    The expectation in the right hand side can be computed with the same technique than for (i), resulting in the series \eqref{symmetric_A/N_ITM}.
    
    \noindent Last, the proof of (iii) is straightforward: letting $k_{\sigma,\nu}\rightarrow 0$ in \eqref{symmetric_A/N_OTM} or \eqref{symmetric_A/N_ITM}, we see that only the terms for $n_1=0$ survive; renaming $n_2:=n$ yields formula \eqref{symmetric_A/N_0} and completes the proof.
    \end{proof}

\paragraph{European option} The European call option pays $S_T-K$ at maturity on the condition that the spot price is greater than the strike:
\begin{equation}
    \mathcal{P}_{eur} (S_T,K) \, := [ S_T-K ] ^+ 
    .
\end{equation}
\begin{formula}[European call] 
    \label{formula:symmetric_European}
    The value at time $t$ of a European call option is: \\
    (i) (OTM price) If $ k_{\sigma,\nu} < 0$, 
        \begin{multline}\label{symmetric_European_OTM}
            C_{eur}^-(k_{\sigma,\nu},\sigma_\nu) \, = \, \frac{F}{2\Gamma(\frac{\tau}{\nu})} \sum\limits_{\substack{n_1 = 0 \\ n_2 = 1}}^{\infty} \frac{(-1)^{n_1}}{n_1!}
            \left[
            \frac{\Gamma(\frac{-n_1+n_2+1}{2} + \alpha)}{\Gamma(\frac{-n_1+n_2}{2}+1)} 
            \left( \frac{-k_{\sigma,\nu}}{\sigma_\nu} \right)^{n_1} \sigma_\nu^{n_2} \right.
            \\
            \left.
            +
            \, 2  \, \frac{\Gamma(-2n_1-n_2-1 -2\alpha)}{\Gamma(-n_1+\frac{1}{2}-\alpha)} 
            \left( \frac{-k_{\sigma,\nu}}{\sigma_\nu} \right)^{2n_1 +1 +2\alpha} (-k_{\sigma,\nu})^{n_2}
            \right]
            .
        \end{multline}
    (ii) (ITM price) If $ k_{\sigma,\nu} > 0$,
        \begin{equation}\label{symmetric_European_ITM}
            C_{eur}^+(k_{\sigma,\nu},\sigma_\nu) \, = \, S - Ke^{-r\tau} -  C_{eur}^-(k_{\sigma,\nu},-\sigma_\nu)
            .
        \end{equation}
    (iii) (ATM price) If $ k_{\sigma,\nu} = 0 $,
        \begin{equation}\label{symmetric_European_0}
            C_{eur}^-(k_{\sigma,\nu},\sigma_\nu) \, = \,  C_{eur}^+(k_{\sigma,\nu},\sigma_\nu) \, = \, \frac{F}{2\Gamma(\frac{\tau}{\nu})} \sum\limits_{n=1}^\infty \frac{\Gamma(\frac{n+1}{2}+\alpha)}{\Gamma(\frac{n}{2}+1)} \, \sigma_\nu^n 
            .
        \end{equation} 
\end{formula}
\begin{proof}
    To prove (i), we first note that, using notations \eqref{moneyness}, we have:
    \begin{equation}
        \mathcal{P}_{eur} (Se^{(r+\omega_{\sigma,\nu})\tau+x},K) \, = \, K (e^{k_{\sigma,\nu}+x} -1) \mathbbm{1}_{ \{x > -k_{\sigma,\nu}\} }
        .
    \end{equation}
    Then, we use the Mellin-Barnes representation (see table \ref{tab:Mellin} in appendix \ref{app:Mellin}, or \cite{Bateman54} p. 313):
    \begin{equation}\label{MB_payoff_european}
        e^{k_{\sigma,\nu}+x} - 1 \, = \, \int\limits_{c_2 - i\infty}^{c_2 + i\infty} (-1)^{-s_
        2} \Gamma(s_2) (k_{\sigma,\nu} + x)^{-s_2} \, \frac{\ud s_2}{2i\pi} \hspace{1cm} (-1 < c_2 < 0)
    \end{equation}
    and we proceed the same way than for proving (i) in Formula \ref{formula:symmetric_A/N}; note that the $n_2$-summation in \eqref{symmetric_European_OTM} now starts in $n_2=1$ instead of $n_2=0$, because the strip of convergence of \eqref{MB_payoff_european} is reduced to $<-1,0>$ instead of $<0,\infty>$ in \eqref{MB_payoff_a/n}.
    
    \noindent To prove (ii), we write
    \begin{equation}
        \mathbb{E}_t^\mathbb{Q} [ (S_T - K) \mathbbm{1}_{ \{ S_T > K \} } ] \, = \, Se^{r\tau} \, - \,  K \,  \, - \, \mathbb{E}_t^\mathbb{Q} [ (S_T -K) \mathbbm{1}_{ \{ S_T < K \} } ]   
    \end{equation}
    and compute the expectation in the r.h.s. following the same technique than for proving (i)
    
    \noindent To prove (iii), we proceed like in the proof of formula~\ref{formula:symmetric_A/N} by letting $k_{\sigma,\nu}\rightarrow 0$ in both OTM and ITM prices \eqref{symmetric_European_OTM} and \eqref{symmetric_European_ITM}.
\end{proof}
An interesting particular case of formula~\eqref{symmetric_European_OTM} occurs when $S$ is \textit{at-the-money forward}, i.e. when $S=F$ or, equivalently, $k=0$. In this case, $k_{\sigma,\nu} = \omega_{\sigma,\nu}\tau$ and, Taylor expanding the martingale adjustment \eqref{symmetric_omega} at first order for small $\sigma$:
\begin{equation}\label{symmetric_omega_approx}
    \omega_{\sigma,\nu} \, \sim \, -\frac{\sigma^2}{2}
\end{equation}
we obtain:
\begin{formula}[European call - At-the-money forward price] 
    \label{formula:symmetric_European_ATM}
    If $S=F$, then under the martingale adjustment approximation \eqref{symmetric_omega}, the value at time $t$ of a European call option is:
    \begin{multline}\label{symmetric_European_ATM}
        C_{eur}^{ATM}(S,r,\tau,\sigma,\nu) \, = \, \frac{S}{2\Gamma(\frac{\tau}{\nu})} \sum\limits_{\substack{n_1 = 0 \\ n_2 = 1}}^{\infty} \frac{(-1)^{n_1}}{n_1!}
            \left[
            \frac{\Gamma(\frac{-n_1+n_2+1}{2} + \alpha)}{\Gamma(\frac{-n_1+n_2}{2}+1)} 
            \left( \frac{\sigma\tau}{\sqrt{2\nu}} \right)^{n_1} \left(\sigma\sqrt{\frac{\nu}{2}}\right)^{n_2} \right.
            \\
            \left.
            +
            \, 2  \, \frac{\Gamma(-2n_1-n_2-1 -2\alpha)}{\Gamma(-n_1+\frac{1}{2}-\alpha)} 
            \left( \frac{\sigma\tau}{\sqrt{2\nu}} \right)^{2n_1 +1 +2\alpha} \left(\frac{\sigma^2\tau}{2}\right)^{n_2}
            \right]
            .
    \end{multline}
\end{formula}
A good estimate of \eqref{symmetric_European_ATM} is given by the leading term of the series, i.e.:
\begin{equation}\label{symmetric_European_ATM_leading}
    C_{eur}^{ATMF}(S,r,\tau,\sigma,\nu) \, \simeq \, \frac{S}{\sqrt{2 \pi}} \, \frac{\Gamma(\frac{1}{2} + \frac{\tau}{\nu})}{\Gamma(\frac{\tau}{\nu})} \, \sigma \sqrt{\nu}
    .
\end{equation}
Note that it follows from Stirling\textquoteright s approximation for the Gamma function that 
\begin{equation}
    \frac{\Gamma \left( \frac{1}{2} + \frac{\tau}{\nu} \right)}{\Gamma \left( \frac{\tau}{\nu} \right)} \, \underset{\nu\rightarrow 0}{\sim} \, \sqrt{\frac{\tau}{\nu}}
    ,
\end{equation}
and therefore, in the low variance limit, the VG price \eqref{symmetric_European_ATM_leading} recovers the approximation of \cite{Brenner94} for the at-the-money-forward Black-Scholes price:
\begin{equation}\label{call_eur_ATMF_BS}
    C_{eur}^{ATMF}(S,r,\tau,\sigma,\nu) \, \simeq \, \frac{S}{\sqrt{2\pi}} \, \sigma \, \sqrt{\tau}
    ,
\end{equation}
that market practitioners often write $C\simeq 0.4 S \sigma\sqrt{\tau}$ because $1/\sqrt{2\pi}\simeq 0.399$. Let us also note that \eqref{symmetric_European_ATM_leading} allows for an estimate of the implied ATM-forward volatility:
\begin{equation}\label{symmetric_European_ATM_IV_approx}
    \sigma_I \, \simeq \, \sqrt{\frac{2\pi}{\nu}} \frac{\Gamma(\frac{\tau}{\nu})}{\Gamma(\frac{1}{2} + \frac{\tau}{\nu})} \, \frac{C_0}{S}
\end{equation}
where $C_0$ is the observed European ATM-forward market price. In the low variance regime (small $\nu$), \eqref{symmetric_European_ATM_IV_approx} becomes:
\begin{equation}
    \sigma_I \, \underset{\nu\rightarrow 0}{\longrightarrow} \, \, \sqrt{\frac{2\pi}{\tau}} \, \frac{C_0}{S}
    .
\end{equation}

\paragraph{Digital option (cash-or-nothing)}
The payoff of the cash-or-nothing call option is
\begin{equation}
    \mathcal{P}_{c/n}(S_T,K) \, = \, \mathbbm{1}_{ \{ S_T>K \} }
\end{equation}
and therefore the option price itself is:
\begin{equation}
    C_{c/n}(k_{\sigma,\nu},\sigma_\nu) \, = \, \frac{1}{K} \, \left(  C_{a/n}(k_{\sigma,\nu},\sigma_\nu) - C_{eur}(k_{\sigma,\nu},\sigma_\nu)  \right)
    .
\end{equation}
Using formulas~\ref{formula:symmetric_A/N} and \ref{formula:symmetric_European}, it is immediate to see that:

\begin{formula}[Cash-or-nothing call] 
    \label{formula:symmetric_C/N}
    The value at time $t$ of a cash-or-nothing call option is: \\
    (i) (OTM price) If $ k_{\sigma,\nu} < 0$, 
    \begin{multline}\label{symmetric_C/N_OTM}
        C_{c/n}^-(k_{\sigma,\nu},\sigma_\nu) \, = \, \frac{e^{-r\tau}}{2\Gamma(\frac{\tau}{\nu})} \sum\limits_{n=0}^{\infty} \frac{(-1)^{n}}{n!}
        \left[
        \frac{\Gamma(\frac{-n+1}{2} + \alpha)}{\Gamma(\frac{-n}{2}+1)} 
        \left( \frac{-k_{\sigma,\nu}}{\sigma_\nu} \right)^{n} \right.
        \\
        \left.
        +
        \, 2  \, \frac{\Gamma(-2n-1 -2\alpha)}{\Gamma(-n+\frac{1}{2}-\alpha)} 
        \left( \frac{-k_{\sigma,\nu}}{\sigma_\nu} \right)^{2n +1 +2\alpha} 
        \right]
        .
    \end{multline}
    (ii) (ITM price) If $ k_{\sigma,\nu} > 0$,
        \begin{equation}\label{symmetric_C/N_ITM}
            C_{c/n}^+(k_{\sigma,\nu},\sigma_\nu) \, = \, e^{-r\tau} - C_{c/n}^-(k_{\sigma,\nu},-\sigma_\nu)
            .
        \end{equation}
    (iii) (ATM price) If $ k_{\sigma,\nu} = 0 $,
        \begin{equation}\label{symmetric_C/N_0}
            C_{eur}^-(k_{\sigma,\nu},\sigma_\nu) \, = \,  C_{eur}^+(k_{\sigma,\nu},\sigma_\nu) \, = 
            \, \frac{e^{-r\tau}}{2}
            .
        \end{equation} 
\end{formula}

\subsection{Other payoffs}\label{subsec:symmetric_other_payoffs}
In this subsection, we show how the pricing formulas established in subsection~\ref{subsec:symmetric_payoffs} can be extended to other exotic options.

\paragraph{Gap option}
A Gap (sometimes called Pay-Later) call has the following payoff:
\begin{equation}\label{Payoff_Gap}
    \mathcal{P}_{gap} (S_T,K_1,K_2) \, = \, (S_T-K_1) \mathbbm{1}_{ \{ S_T > K_2 \} }
\end{equation}
and degenerates into the European call when trigger and strike prices coincide ($K_1=K_2=K$). From the definition \eqref{Payoff_Gap}, it is immediate to see that the value at time $t$ of the Gap call is:
\begin{equation}
    C_{gap} (k_{\sigma,\nu},\sigma_\nu) = C_{a/n}(k_{\sigma,\nu},\sigma_\nu) \, - \, K_1 \, C_{c/n} (k_{\sigma,\nu},\sigma_\nu)
\end{equation}
where the value of the asset-or-nothing and cash-or-nothing calls are given by formulas~\ref{formula:symmetric_A/N} and \ref{formula:symmetric_C/N} for $K=K_2$.

\paragraph{Power option}
Power options deliver a higher payoff than vanilla options, and are used to increase the leverage ratio of trading strategies. For instance, an asset-or-nothing power call has the non-linear payoff:
\begin{equation}
    \mathcal{P}_{pow.a/n}(S_T,K) \, = \, S_T^q \mathbbm{1}_{ \{S_T^q > K\} }
\end{equation}
for some $q\geq1$. We can remark that
\begin{equation}\label{Payoff_power_a/n}
    \mathcal{P}_{pow.a/n}(Se^{(r+\omega_{\sigma,\nu})\tau+x},K) \, = \, Ke^{q ( \tilde{k}_{\sigma,\nu} +x) } \mathbbm{1}_{ \{ x > -\tilde{k}_{\sigma,\nu} \} }
\end{equation}
where
\begin{equation}
    \tilde{k}_{\sigma,\nu} \, := \, \log\frac{S}{K^{\frac{1}{q}}} + r\tau + \omega_{\sigma,\nu}
    .
\end{equation}
Introducing a Mellin-Barnes representation for the exponential term in \eqref{Payoff_power_a/n} like we did in eq. \eqref{MB_payoff_a/n} shows that the pricing formula for the OTM ($\tilde{k}_{\sigma,\nu} < 0$) power asset-or-nothing call is the same than \eqref{symmetric_A/N_OTM}, when replacing $k_{\sigma,\nu}$ by $\tilde{k}_{\sigma,\nu}$ and multiplying out the series terms by $q^{n_2}$, that is:
\begin{multline}\label{symmetric_power_A/N_OTM}
            C_{pow.a/n}^-(\tilde{k}_{\sigma,\nu},\sigma_\nu) \, = \, \frac{F}{2\Gamma(\frac{\tau}{\nu})} \sum\limits_{\substack{n_1 = 0 \\ n_2 = 0}}^{\infty} \frac{(-1)^{n_1}}{n_1!} q^{n_2}
            \left[
            \frac{\Gamma(\frac{-n_1+n_2+1}{2} + \alpha)}{\Gamma(\frac{-n_1+n_2}{2}+1)} 
            \left( \frac{-\tilde{k}_{\sigma,\nu}}{\sigma_\nu} \right)^{n_1} \sigma_\nu^{n_2} \right.
            \\
            \left.
            +
            \, 2  \, \frac{\Gamma(-2n_1-n_2-1 -2\alpha)}{\Gamma(-n_1+\frac{1}{2}-\alpha)} 
            \left( \frac{-\tilde{k}_{\sigma,\nu}}{\sigma_\nu} \right)^{2n_1 +1 +2\alpha} (-\tilde{k}_{\sigma,\nu})^{n_2}
            \right]
            .
        \end{multline}
Letting $\tilde{k}_{\sigma,\nu}\rightarrow 0$ in \eqref{symmetric_power_A/N_OTM} yields the ATM price:
\begin{equation}\label{symmetric_power_A/N_ATM}
    C_{pow.a/n}^{atm}(k_{\sigma,\nu},\sigma_\nu) \, = \, \frac{F}{2\Gamma(\frac{\tau}{\nu})} \sum\limits_{n=0}^\infty \frac{\Gamma(\frac{n+1}{2}+\alpha)}{\Gamma(\frac{n}{2}+1)} \, (q \sigma_\nu)^n 
    .
\end{equation}
The ITM price can be obtained from the parity relation:
\begin{equation}\label{symmetric_power_A/N_ITM} 
    C_{pow.a/n}^+(\tilde{k}_{\sigma,\nu},\sigma_\nu) \, = \, S^q e^{(q-1)(r-q \omega_{\sigma,\nu})\tau} \, - \, C_{pow.a/n}^-(\tilde{k}_{\sigma,\nu},-\sigma_\nu)
\end{equation}
which follows from the risk-neutral expectation of $S^q$. Extension of formulas \eqref{symmetric_power_A/N_OTM}, \eqref{symmetric_power_A/N_ATM} and \eqref{symmetric_power_A/N_ITM} to European and cash-or-nothing options are straightforward.

\paragraph{Log option}
Log options were introduced in \cite{Wilmott06} and are basically options on the rate of return of the underlying asset. The call's payoff is:
\begin{equation}
    \mathcal{P}_{log} (S_T,K) \, := 
    \left[ \log S_T \, - \, \log K  \right]^+
\end{equation}
for $K>0$. Let us remark that we have:
\begin{equation}
    \mathcal{P}_{log} (Se^{(r+\omega_{\sigma,\nu})\tau+x},K) \, = \,  [ k_{\sigma,\nu} +x  ]^+
\end{equation}
and therefore, the $P_{\sigma,\nu}^*$ function can be written down as:
\begin{equation}
    P_{\sigma,\nu}^* (s_1) \, = \, \frac{(-k_{\sigma,\nu})^{-s_1+2+\alpha}}{(s_1-2-\alpha)(s_1-1-\alpha)}
\end{equation}
on the condition that $k_{\sigma,\nu}<0$. Inserting in the pricing formula \eqref{Risk-neutral:convolution} and summing up all arising residues yields the OTM log call price:
\begin{multline}\label{symmetric_log_OTM} 
    C_{log}^- (k_{\sigma,\nu},\sigma_\nu) \, = \, \frac{e^{-r\tau}}{\sqrt{\pi}\Gamma(\frac{\tau}{\nu})} \left[ \Gamma(1+\alpha)\sigma_\nu \, + \, \frac{\sqrt{\pi}}{2}\Gamma(\frac{1}{2}+\alpha)k_{\sigma,\nu} \, + \,  
    \right.
    \\
    \left.
    \sum\limits_{n=0}^\infty \frac{(-1)^{n}}{n!} \left(
    \frac{\Gamma(\alpha-n)}{(2n+2)(2n+1)} (-k_{\sigma,\nu})^{2n+2}(2\sigma_\nu)^{-2n-1}
    + \frac{\Gamma(-\alpha-n)}{( 2n + 2\alpha +2)(2n+2\alpha+1)} 
    \right.\right.
    \\
    \left.\left.
    (-k_{\sigma,\nu})^{2n+2\alpha+2}(2\sigma_\nu)^{-2\alpha-2n-1}
    \right)\right]
    .
\end{multline}
Letting $k_{\sigma,\nu}\rightarrow 0$ in \eqref{symmetric_log_OTM} yields the ATM price:
\begin{equation}
    C_{log}^{atm} (k_{\sigma,\nu},\sigma_\nu) \, = \, \frac{e^{-r\tau}}{\sqrt{\pi}} \frac{\Gamma(\frac{1}{2} + \frac{\tau}{\nu})}{\Gamma(\frac{\tau}{\nu})} \sigma_\nu
    .
\end{equation}


\section{Extension to the asymmetric process}\label{sec:asymmetricVG}

In this section, we establish a factorized pricing formula for the asymmetric VG process ($\theta\neq 0$) and apply the result to the cash-or-nothing option. 

\subsection{Pricing formula}\label{subsec:asymmetric_pricing_formula}

Let $\sigma,\nu > 0$, and let $\theta\in\mathbb{R}$. Let us define the positive quantity
\begin{equation}
    q_{\sigma,\nu,\theta} \, := \, \frac{1}{2\sigma^2\nu} \, + \, \left( \frac{\theta}{2\sigma^2}  \right)^2
\end{equation}
and let $F^*_\nu$ be like in lemma \eqref{F_nu}. Introducing a Mellin representation for the Bessel function like in the proof of lemma~\ref{lemma:density_asymmetricVG}, and a supplementary Mellin representation for the exponential term $e^{\theta\frac{x}{\sigma^2}}$ in \eqref{VG_density_2} yields:
 
\begin{lemma}\label{lemma:density_asymmetricVG}
    Let $\mathrm{P}_1 :=( | \frac{\tau}{\nu} - \frac{1}{2} | ,\infty) \times (0,\infty) $; then for
    any $ (c_1,c_2) \in \mathrm{P}_1 $, the following Mellin-Barnes representation holds true:
    \begin{multline}\label{Density_asymmetricVG_MB}
        f_{\sigma,\nu,\theta}(x,\tau) \, = \, \frac{(\frac{2\sigma^2}{\nu}+\theta^2)^{-\frac{\alpha}{2}}}{ 2 \sqrt{2\pi} \sigma\nu^{\frac{\tau}{\nu}} \Gamma(\frac{\tau}{\nu}) } \, \times
        \\
        \int\limits_{c_1 - i\infty}^{c_1 + i\infty} \int\limits_{c_2 - i\infty}^{c_2 + i\infty} \, (-1)^{-s_2} \Gamma(s_2) F_{\nu}^*(s_1) \, q_{\sigma,\nu,\theta}^{-\frac{s_1}{2}} \left(  \frac{\theta}{\sigma^2} \right)^{-s_2}
        x^{-s_2} |x|^{-s_1+ \frac{\tau}{\nu} -\frac{1}{2} }
        \frac{\ud s_1 \ud s_2}{(2i\pi)^2}
        .
    \end{multline}
\end{lemma}

Let us now introduce the Mellin-like transform of the payoff function:
\begin{equation}\label{asymmetric_Payoff_transform}
    P_{\sigma,\nu,\theta}^*(s_1,s_2) \, = \, \Gamma(s_2) \, \int\limits_{-\infty}^{\infty} \, \mathcal{P} \left( S e^{(r+\omega_{\sigma,\nu,\theta})\tau +x} , \underline{K}  \right) \, x^{-s_2} |x|^{-s_1+ \frac{\tau}{\nu} -\frac{1}{2} } \, \ud x
\end{equation}
and assume that it converges for $Re(s_1,s_2)\in\mathrm{P}_2$ for a certain subset $\mathrm{P}_2\subset\mathbb{R}^2$. As before, let us assume that $\alpha:=\frac{\tau}{\nu} - \frac{1}{2}\notin\mathbb{Z}$ and introduce the supplementary notation $\theta_\sigma:=\frac{\theta}{\sigma^2}$, so that we have:
\begin{equation}
    q_{\sigma,\nu,\theta} \, = \, \frac{1}{4} \left( \frac{1}{\sigma_\nu^2} \, + \, \theta_\sigma^2  \right)
\end{equation}
where $\sigma_\nu$ is the one defined in \eqref{asymmetricVG_notations}. As a consequence of lemma~\ref{lemma:density_symmetricVG} and of the pricing formula \eqref{Risk-neutral_2}, we have the following factorized formula:

\begin{proposition}[Factorization in the Mellin space]
    \label{prop:asymmetric_factorization}
    Let $(c_1,c_2)\in \mathrm{P}$ where $
    \mathrm{P} := \mathrm{P}_1 \cap \mathrm{P}_2
    $ is assumed to be nonempty. Then the value at time $t$ of a containgent claim $\mathcal{C}$ with maturity $T$ and payoff $\mathcal{P}(S_T,\underline{K})$ is equal to:
    \begin{multline}\label{asymmetric_Risk-neutral:convolution}
        \mathcal{C} (S,\underline{K},r,\tau,\sigma,\nu,\theta) \, = \, 
        \frac{e^{-r\tau}}{ 2^{2+2\alpha} \sqrt{\pi} \sigma_\nu^{1+2\alpha} \Gamma(\frac{\tau}{\nu}) } \, \times
        \\
        \int\limits_{c_1-i\infty}^{c_1+i\infty}\int\limits_{c_2-i\infty}^{c_2+i\infty} \, 
        P_{\sigma,\nu,\theta}^*(s_1,s_2) F_{\nu}^*(s_1)
        \, \left( -\theta_\sigma \right)^{-s_2}  \left( q_{\sigma,\nu,\theta} \right)^{-\frac{s_1+\alpha}{2}}  \, \frac{\ud s_1 \ud s_2}{(2i\pi)^2}
        .
    \end{multline}
\end{proposition}
Let us now demonstrate how to implement proposition~\ref{prop:asymmetric_factorization} on the example of the cash-or-nothing option; we will denote the risk-neutral moneyness by
\begin{equation}\label{asymmetric_RN_moneyness}
    k_{\sigma,\nu,\theta} \, := \, k \, + \, \omega_{\sigma,\nu,\theta}\tau
    .
\end{equation}

\paragraph{Example: Digital option (cash-or-nothing)}
\begin{formula}[Cash-or-nothing call] 
    \label{formula:asymmetric_C/N}
    The value at time $t$ of a cash-or-nothing call option is: \\
    (i) (OTM price) If $ k_{\sigma,\nu} < 0$, 
    \begin{align}\label{asymmetric_C/N_OTM}
        & C_{c/n}^-(k_{\sigma,\nu,\theta},q_{\sigma,\nu,\theta},\sigma_\nu,\theta_\sigma) \, = \,      \frac{e^{-r\tau}}{ 2^{2+2\alpha} \sqrt{\pi} \sigma_\nu^{1+2\alpha} \Gamma(\frac{\tau}{\nu})} \, \times \nonumber
        \\
        & \left[ \sum\limits_{n=0}^{\infty} \frac{1}{n!} \Gamma\left(\frac{n+1}{2}\right)\Gamma\left(\frac{n+1}{2} + \alpha\right)
        q_{\sigma,\nu,\theta}^{-\frac{n+1}{2} - \alpha}\theta_\sigma^{n}
        \, +
        2 \sum\limits_{\substack{n_1 = 0 \\ n_2 = 0}}^{\infty} \frac{(-1)^{n_1}}{n_1!n_2!}
        \right. \nonumber
        \\
        & \left.
        \left( \frac{\Gamma(\alpha-n_1)}{-2n_1-n_2-1}(-k_{\sigma,\nu,\theta})^{2n_1+n_2+1}q_{\sigma,\nu,\theta}^{-\alpha+n_1} 
        + \frac{\Gamma(-\alpha-n_1)}{-2n_1-n_2-1-2\alpha}(-k_{\sigma,\nu,\theta})^{2n_1+n_2+1+2\alpha}q_{\sigma,\nu,\theta}^{n_1} 
        \right)
        {\theta_\sigma}^{n_2}
        \right]
        .
    \end{align}
    (ii) (ITM price) If $ k_{\sigma,\nu} > 0$,
        \begin{equation}\label{asymmetric_C/N_ITM}
        C_{c/n}^+(k_{\sigma,\nu,\theta},q_{\sigma,\nu,\theta},\sigma_\nu,\theta_\sigma) \, = \, e^{-r\tau} \, - \, C_{c/n}^-(-k_{\sigma,\nu,\theta},q_{\sigma,\nu,\theta},\sigma_\nu,-\theta_\sigma)
        .
        \end{equation}
    (iii) (ATM price) If $ k_{\sigma,\nu} = 0 $,
        \begin{equation}\label{asymmetric_C/N_0}
        C_{c/n}^{atm}(q_{\sigma,\nu,\theta},\sigma_\nu,\theta_\sigma) \, = \, 
        \frac{e^{-r\tau}}{ 2^{2+2\alpha} \sqrt{\pi} \sigma_\nu^{1+2\alpha} \Gamma(\frac{\tau}{\nu})}\sum\limits_{n=0}^{\infty} \frac{1}{n!} \Gamma\left(\frac{n+1}{2}\right)\Gamma\left(\frac{n+1}{2} + \alpha\right)
        q_{\sigma,\nu,\theta}^{-\frac{n+1}{2} - \alpha}\theta_\sigma^{n}
        .
        \end{equation} 
\end{formula}

\begin{proof}
    To prove (i), we use notation \eqref{asymmetric_RN_moneyness} for the moneyness to write the cash-or-nothing payoff as:
    \begin{equation}
        \mathcal{P}_{c/n} (Se^{(r+\omega_{\sigma,\nu,\theta})\tau+x},K) \, = \,  \mathbbm{1}_{ \{x\geq-k_{\sigma,\nu,\theta}\} }
    \end{equation}
    and therefore the $P^*_{\sigma,\nu,\theta}(s_1,s_2)$ function \eqref{asymmetric_Payoff_transform} is:
    \begin{equation}
        P^*_{\sigma,\nu,\theta}(s_1,s_2) \, = \, \Gamma(s_2) \, \int\limits_{-k_{\sigma,\nu,\theta}}^\infty \, x^{-s_1-s_2+\alpha} \, \ud x 
        \, = \, \Gamma(s_2) \frac{(-k_{\sigma,\nu,\theta})^{-s_1-s_2+1+\alpha}}{s_1+s_2-1-\alpha} 
    \end{equation}
    and exists as $k_{\sigma,\nu,\theta}<0$. Inserting in the pricing formula \eqref{asymmetric_Risk-neutral:convolution}, the OTM cash-or-nothing call thus reads:
    \begin{multline}\label{asymmetric_C/N_OTM_integral}
        \mathcal{C}_{c/n} (S,\underline{K},r,\tau,\sigma,\nu,\theta) \, = \, 
        \frac{e^{-r\tau}}{ 2^{2+2\alpha} \sqrt{\pi} \sigma_\nu^{1+2\alpha} \Gamma(\frac{\tau}{\nu}) } \, \times
        \\
        \int\limits_{c_1-i\infty}^{c_1+i\infty}\int\limits_{c_2-i\infty}^{c_2+i\infty} \, 
        \frac{\Gamma(\frac{s_1- \alpha}{2} ) \Gamma( \frac{s_1 + \alpha}{2} ) \Gamma(s_2) }{s_1+s_2-1-\alpha}
        \, (-k_{\sigma,\nu,\theta})^{-s_1-s_2+1+\alpha} \left( -\theta_\sigma \right)^{-s_2}  \left( q_{\sigma,\nu,\theta} \right)^{-\frac{s_1+\alpha}{2}}  \, \frac{\ud s_1 \ud s_2}{(2i\pi)^2}
    \end{multline}
where $(c_1,c_2)$ belongs to the convergence polyhedron of \eqref{asymmetric_C/N_OTM_integral} $\mathrm{P}:=\{ (s_1,s_2)\in\mathbb{C}^2, Re(s_1) > |\alpha|, Re(s_2) > 0 \}$. The associated characteristic vector is 
\begin{equation}
    \Delta \, = \,
    \begin{bmatrix}
        1\\1
    \end{bmatrix}
\end{equation}
and therefore, the cone $\Pi_:=\{ (s_1,s_2)\in\mathbb{C}^2, Re(s_1) < |\alpha|, Re(s_2) < 0 \}$ is located in the admissible region and is compatible with the divisors of \eqref{asymmetric_C/N_OTM_integral}; thus, from the residue theorem \eqref{Tsikh_residue}, the integral equals the sum of the residues of its integrand in $\Pi$. There are three series of poles:
    \begin{itemize}
        \item[-] Series 1: when $\Gamma(s_2)$ is singular and $s_1+s_2-1-\alpha=0$, that is when $(s_1,s_2)=(n+1+\alpha,-n)$, $n\in\mathbb{N}$, resulting in the first series in \eqref{asymmetric_C/N_OTM};
        \item[-] Series 2: when $\Gamma(\frac{s_1- \alpha}{2} ) $ and $\Gamma(s_2)$ are singular, that is when $(s_1,s_2) = (\alpha-2n_1,-n_2)$, $n_1,n_2\in\mathbb{N}$;
        \item[-] Series 3: when $\Gamma(\frac{s_1 + \alpha}{2} ) $ and $\Gamma(s_2)$ are singular, that is when $(s_1,s_2) = (-\alpha-2n_1,-n_2)$,  $n_1,n_2\in\mathbb{N}$; series 2 and 3 result in the second series in \eqref{asymmetric_C/N_OTM}.
    \end{itemize}

\noindent To prove (ii), we write 
\begin{equation}
    \mathbb{E}_t^{\mathbb{Q}} [\mathbbm{1}_{\{S_T>K\}}] \, = \, 1 \, - \, \mathbb{E}_t^{\mathbb{Q}} [\mathbbm{1}_{\{S_T<K\}}]
\end{equation}
and apply the same technique than in (i) to the expectation in the right hand side.

\noindent Last, (iii) results from letting $k_{\sigma,\nu,\theta}\rightarrow 0$ in \eqref{asymmetric_C/N_OTM} or \eqref{asymmetric_C/N_ITM}.
\end{proof}

\section{Comparisons with numerical results}\label{sec:tests}

We now test the series formulas obtained in the paper for the European and digital options in both the symmetric and asymmetric case, by benchmarking the results with various numerical techniques.

\subsection{Fourier and FFT techniques}\label{subsec:Fourier}

In this section, we focus on comparing our results with some classic Fourier-related techniques. First, let us recall that, following \cite{Lewis01}, digital option prices admit convenient representation involving the risk-neutral characteristic function and the log-forward moneyness; the asset-or-nothing call can be written as
\begin{equation}\label{Lewis_A/N}
    C_{a/n}(S,K,r,\tau,\sigma,\nu,\theta) \, = \, S \, \left( \, \frac{1}{2} \,+ \, \frac{1}{\pi}\int\limits_0^{\infty} \, Re \, \left[ \frac{e^{i u k} \tilde{\Phi}_{\sigma,\nu,\theta}(u-i,\tau)}{iu} \right] \ud u \, \right)
    ,
\end{equation}
and the cash-or-nothing call as
\begin{equation}\label{Lewis_C/N}
    C_{c/n}(S,K,r,\tau,\sigma,\nu,\theta) \, = \, e^{-r\tau} \, \left( \, \frac{1}{2} \,+ \, \frac{1}{\pi}\int\limits_0^{\infty} \, Re \, \left[ \frac{e^{i u k} \tilde{\Phi}_{\sigma,\nu,\theta}(u,\tau)}{iu} \right] \ud u \, \right)
    ,
\end{equation}
where $k:=\log\frac{S}{K}+r\tau$ as before, and where the characteristic function \eqref{characteristic_function} has been normalized by the martingale adjustment
\begin{equation}
    \tilde{\Phi}_{\sigma,\nu,\theta} (u,t) \, := \, e^{i u t \omega_{\sigma,\nu,\theta} } \Phi_{\sigma,\nu,\theta}(u,t)
\end{equation}
so that the martingale condition $\tilde{\Phi}_{\sigma,\nu,\theta} (-i,t)=1$ holds true. Second, regarding European options, we will use the representation given in \cite{Carr99} based on the introduction of a dampling factor $a$ to avoid the divergence in $u=0$; namely, let
\begin{equation}
    \varphi_{\sigma,\nu,\theta}(u,t) \, := \, e^{i u [\log S + (r + \omega_{\sigma,\nu,\tau}) t]} \, \Phi_{\sigma,\nu,\theta}(u,t) 
    ,
\end{equation}
then the European option price admits the representation:
\begin{equation}\label{CarrMadan}
    C_{eur}(S,K,r,\tau,\sigma,\nu,\theta) \, = \, \frac{e^{-a\log K - r\tau}}{\pi} \, \int\limits_0^{\infty} \,
    e^{-i u \log K}
    Re \left[  \frac{\varphi_{\sigma,\nu,\theta}(u-(a+1)i,\tau)}{a^2+a-u^2+i(2a+1)u}  \right]
    \, \ud u
    ,
\end{equation}
where $a<0<a_{max}$, and $a_{max}$ is determined by the square integrability condition $\varphi_{\sigma,\nu,\theta}(-(a+1)i,\tau)<\infty$; in the symmetric VG model,
\begin{equation}
    a_{max} \, = \, \frac{1}{\sigma}\sqrt{\frac{2}{\nu}} \, - \, 1
    .
\end{equation}
Integrals \eqref{Lewis_A/N}, \eqref{Lewis_C/N} and \eqref{CarrMadan} can be carried out very easily via a classical recursive algorithm after truncating the integration region (or, in the case of \eqref{CarrMadan}, by a Fast Fourier Transform algorithm); typically, restraining the Fourier variable to $u\in[0,10^4]$ is sufficient to obtain an excellent level of precision when the time to maturity is not too small; however, as we will see, in the limit $\tau\rightarrow 0$ the numerical evaluation of \eqref{CarrMadan} has more difficulties to converge if the option is deep OTM, and in that case the series in formula~\ref{formula:symmetric_European} constitute a more efficient tool.

\paragraph{Symmetric process}
In table~\ref{tab:symmetric_C/N}, we compare the prices obtained by the series in formula~\ref{formula:symmetric_C/N} truncated at $n=n_{max}$, with a numerical evaluation of \eqref{Lewis_C/N}; we investigate several market configurations, from ITM to OTM options, for long and short maturities.
\begin{table}[ht]
 \caption{Cash-or-nothing prices given by formula~\ref{formula:symmetric_C/N} for various truncations and different market configurations, vs. numerical evaluation of the integral \eqref{Lewis_C/N}. Parameters: $K=4000$, $r=1\%$, $\sigma=0.2$, $\nu=0.85$.}
 \label{tab:symmetric_C/N}       
 \centering
 \begin{tabular}{cccccc}
 \hline\noalign{\smallskip}
 & $n_{max}=3$ & $n_{max}=5$ & $n_{max}=10$ & $n_{max}=15$ & Lewis \eqref{Lewis_C/N} [$\theta=0$]  \\
 \noalign{\smallskip}\hline\noalign{\smallskip}
 \multicolumn{2}{l}{{\bfseries Long term options} ($\tau=2$)} & & & \\
$S=5000$  & 0.8870 & 0.7885 & 0.7755 & 0.7754 & 0.7754   \\
$S=4200$  & 0.5372 & 0.5373 & 0.5373 & 0.5373 & 0.5373  \\
ATM       & 0.4901 & 0.4901 & 0.4901 & 0.4901 & 0.4901  \\
$S=3800$  & 0.3734 & 0.3739 & 0.3740 & 0.3740 & 0.3740 \\
$S=3000$  &-0.9399 &-0.1439 & 0.1159 & 0.1181 & 0.1181 \\
  \noalign{\smallskip}\hline\noalign{\smallskip}
 \multicolumn{2}{l}{{\bfseries Short term options} ($\tau=0.5$)} & & & \\
$S=5000$  & 0.8415 & 0.9355 & 0.9410 & 0.9410 & 0.9410 \\
$S=4200$  & 0.7203 & 0.7104 & 0.7104 & 0.7104 & 0.7104  \\
ATM       & 0.4975 & 0.4975 & 0.4975 & 0.4975 & 0.4975 \\
$S=3800$  & 0.2487 & 0.2486 & 0.2486 & 0.2486 & 0.2486 \\
$S=3000$  & 0.4814 & 0.0731 & 0.0282 & 0.0281 & 0.0281 \\
 \noalign{\smallskip}\hline
\end{tabular}
\end{table}

We perform the same analysis for the asset-or-nothing call in table~\ref{tab:symmetric_A/N}; the series given by formula~\ref{formula:symmetric_A/N} are truncated at $n=m=max$, and compared to a numerical evaluation of \eqref{Lewis_A/N}. 
\begin{table}[ht]
 \caption{Asset-or-nothing prices given by formula~\ref{formula:symmetric_A/N} for various truncations and different market configurations, vs. numerical evaluation of the integral \eqref{Lewis_A/N}. Parameters: $K=4000$, $r=1\%$, $\sigma=0.2$, $\nu=0.85$.}
 \label{tab:symmetric_A/N}       
 \centering
 \begin{tabular}{cccccc}
 \hline\noalign{\smallskip}
 & ${max}=3$ & ${max}=5$ & ${max}=10$ & ${max}=15$ & Lewis \eqref{Lewis_A/N} [$\theta=0$]  \\
 \noalign{\smallskip}\hline\noalign{\smallskip}
 \multicolumn{2}{l}{{\bfseries Long term options} ($\tau=2$)} & & & \\
$S=5000$  & 4758.18 & 4374.28 & 4307.02 & 4306.93 & 4306.93  \\
$S=4200$  & 2739.14 & 2737.53 & 2737.49 & 2737.49 & 2737.49  \\
ATM       & 2472.07 & 2474.66 & 2474.72 & 2474.72 & 2474.72  \\
$S=3800$  & 1851.86 & 1855.46 & 1855.51 & 1855.51 & 1855.51 \\
$S=3000$  &-3390.31 &-284.115 & 560.097 & 568.846 & 568.846 \\
  \noalign{\smallskip}\hline\noalign{\smallskip}
 \multicolumn{2}{l}{{\bfseries Short term options} ($\tau=0.5$)} & & & \\
$S=5000$  & 4192.59 & 4769.19 & 4806.50 & 4806.52 & 4806.50   \\
$S=4200$  & 3168.47 & 3168.75 & 3168.74 & 3168.74 & 3168.74  \\
ATM       & 2196.77 & 2197.07 & 2197.07 & 2197.07 & 2797.07  \\
$S=3800$  & 1113.19 & 1113.80 & 1113.80 & 1113.80 & 1113.80 \\
$S=3000$  & 1265.39 & 222.654 & 127.749 & 127.292 & 127.293  \\
 \noalign{\smallskip}\hline
\end{tabular}
\end{table}

We can remark that both in the cash-or-nothing and the asset-or-nothing cases, the convergence to the option price is very fast; in particular when the underlying is close to the money, it is sufficient to consider only terms up to $n=5$ to get a precision of $10^{-3}$ in the asset-or-nothing price, independently of the maturity. We also note that, as a whole, the convergence goes faster in the ITM region, and is slightly slower when the option is very deep OTM; this is because, for fixed $\tau$, the moneyness is logarithmically driven:
\begin{equation}
    |k_{\sigma,\nu}| \, \underset{S\rightarrow 0}{\sim} \, |\log S|
\end{equation}
and therefore becomes bigger when $S$ decreases, which tends to slow down the convergence of the OTM series \eqref{symmetric_A/N_OTM} and \eqref{symmetric_C/N_OTM}. On the contrary, as $|k_{\sigma,\nu}|$ grows very slowly when $S$ grows, the convergence of the ITM series \eqref{symmetric_A/N_ITM} and \eqref{symmetric_C/N_ITM} remains fast even for deep ITM situations. Last, a remarkable feature is the short maturity behavior of the series; as
\begin{equation}
    |k_{\sigma,\nu}| \, \underset{\tau\rightarrow 0}{\sim} \, \left| \log \frac{S}{K} \right|
\end{equation}
then if the ratio $S/K$ is not too big or too small, the moneyness will be close to $0$ and will trigger an accelerated convergence to the option price. This is a very interesting feature, as the short maturity situation is generally not favorable for numerical Fourier inversion (see \cite{Carr99}; the numerical oscillation for $\tau\rightarrow 0$ leads the authors to introduce a "modified time" approach, involving a multiplication by a hyperbolic sine function instead of the usual Fourier kernel). This situation is displayed in table~\ref{tab:symmetric_small_tau}. Observe in particular that, for deep OTM situations, the numerical inversion fails to converge even when integrating up to $u=10^4$; on the contrary, the series \eqref{symmetric_European_OTM} converges to the price with a precision of $10^{-3}$ when considering only terms up to $n=m=20$ (or even $10$ if the option is not to deep OTM).  
\begin{table}[ht]
 \caption{Short maturity prices, obtained via the OTM series \eqref{symmetric_European_OTM} for the Europan prices truncated at $n=m=max$ vs. numerical evaluations of \eqref{CarrMadan} for various upper bounds $u_{max}$ of the integration region; the series \eqref{symmetric_European_OTM} converges extremely fast to the price, while numerical integration has more difficulty to converge, notably in the deep OTM region. Parameters: $K=4000$, $r=1\%$, $\sigma=0.2$, $\nu=0.85$, $a=1$.}
 \label{tab:symmetric_small_tau}       
 \centering
 \begin{tabular}{lccc}
  \hline\noalign{\smallskip}
  {\bfseries OTM (S=3000)} &  \multicolumn{3}{c}{ Series \eqref{symmetric_European_OTM} }   \\
  & ${max}=5$ & ${max}=10$ & ${max}=20$ \\
 \noalign{\smallskip}\hline\noalign{\smallskip}
 $\tau=\frac{1}{12}$ (1 month)  & 3.953 & 1.802 & 1.802   \\
 $\tau=\frac{1}{52}$ (1 week)   & 0.804 & 0.388 & 0.388  \\
 $\tau=\frac{1}{360}$ (1 day)   & 0.113 & 0.055 & 0.055  \\
 \noalign{\smallskip}
 \hline
  &  \multicolumn{3}{c}{ Carr-Madan \eqref{CarrMadan} }   \\
  & $u_{max}=10^2$ & $u_{max}=10^3$ & $u_{max}=10^4$ \\
 \noalign{\smallskip}\hline\noalign{\smallskip}
 $\tau=\frac{1}{12}$ (1 month)  & 1.870 & 1.803 & 1.802   \\
 $\tau=\frac{1}{52}$ (1 week)   & 0.476 & 0.390 & 0.388  \\
 $\tau=\frac{1}{360}$ (1 day)   & 0.149 & 0.057 & 0.055 \\
 \noalign{\smallskip}
 \hline
 \hline
 \noalign{\smallskip}
   {\bfseries Deep OTM (S=2000)} &  \multicolumn{3}{c}{ Series \eqref{symmetric_European_OTM} }   \\
  & ${max}=10$ & ${max}=20$ & ${max}=30$ \\
 \noalign{\smallskip}\hline\noalign{\smallskip}
 $\tau=\frac{1}{12}$ (1 month)  & 4.9881 & 0.0470 & 0.0470   \\
 $\tau=\frac{1}{52}$ (1 week)   & 0.9093 & 0.0096 & 0.0096  \\
 $\tau=\frac{1}{360}$ (1 day)   & 0.1246 & 0.0013 & 0.0013 \\
 \noalign{\smallskip}
 \hline
  &  \multicolumn{3}{c}{ Carr-Madan \eqref{CarrMadan} }   \\
  & $u_{max}=10^2$ & $u_{max}=10^3$ & $u_{max}=10^4$ \\
 \noalign{\smallskip}\hline\noalign{\smallskip}
 $\tau=\frac{1}{12}$ (1 month)  & 0.0410 & 0.0469 & 0.0467  \\
 $\tau=\frac{1}{52}$ (1 week)   & 0.0031 & 0.0093 & 0.0146  \\
 $\tau=\frac{1}{360}$ (1 day)   &-0.0051 & 0.0009 & 0.0115  \\
 \noalign{\smallskip}
 \hline
\end{tabular}
\end{table}

\paragraph{Asymmetric process}
In table~\ref{tab:asymmetric_C/N}, we compare the prices obtained by the series in formula~\ref{formula:asymmetric_C/N} truncated at $n=m={max}$, with a numerical evaluation of \eqref{Lewis_C/N}, for negative or positive asymmetry parameter $\theta$; again, the convergence is very fast in every situation. Note that the moneyness \eqref{asymmetric_RN_moneyness} is a function of $\theta$; for the set of parameters of table~\ref{tab:symmetric_C/N}, this results in two distinct ATM prices:
\begin{equation}
    \left\{
    \begin{aligned}
        S_{ATM}^+ \,= \, Ke^{-(r+\omega_{\sigma,\nu,\theta})\tau} \, = \, 5050.24 \hspace*{0.5cm} [\theta=+0.1] \\
        S_{ATM}^- \,= \, Ke^{-(r+\omega_{\sigma,\nu,\theta})\tau} \, = \, 3358.52 \hspace*{0.5cm} [\theta=-0.1]   
        .
    \end{aligned}
     \right.
\end{equation}

\begin{table}[ht]
 \caption{Cash-or-nothing prices given by formula~\ref{formula:asymmetric_C/N} for various truncations and positive or negative asymmetry, vs. numerical evaluation of the integral \eqref{Lewis_C/N}. Parameters: $K=4000$, $r=1\%$, $\tau=2$, $\sigma=0.2$, $\nu=0.85$.}
 \label{tab:asymmetric_C/N}       
 \centering
 \begin{tabular}{cccccc}
 \hline\noalign{\smallskip}
 & ${max}=3$ & ${max}=5$ & ${max}=10$ & ${max}=15$ & Lewis \eqref{Lewis_C/N} [$\theta=+0.1$]  \\
 \noalign{\smallskip}\hline\noalign{\smallskip}
$S=6000$  & 0.9096 & 0.9008 & 0.8993 & 0.8993 & 0.8993  \\
ATM       & 0.7064 & 0.7257 & 0.7288 & 0.7288 & 0.7288  \\
$S=3000$  &-15.545 &-0.3877 & 0.1364 & 0.1364 & 0.1364 \\
 \hline\noalign{\smallskip}
 & ${max}=3$ & ${max}=5$ & ${max}=10$ & ${max}=15$ & Lewis \eqref{Lewis_C/N} [$\theta=-0.1$]  \\
 \noalign{\smallskip}\hline\noalign{\smallskip}
$S=5000$  & 1.7747 & 0.7749 & 0.7605 & 0.7605 & 0.7605  \\
ATM       & 0.2412 & 0.2500 & 0.2514 & 0.2514 & 0.2514  \\
$S=2000$  &-1.1760 &-0.0400 & 0.0047 & 0.0047 & 0.0047 \\
 \noalign{\smallskip}\hline
\end{tabular}
\end{table}

In table~\ref{tab:asymmetric_small_tau} we examinate the case $\tau\rightarrow 0$; like for the symmetric model, the convergence of formula~\ref{formula:asymmetric_C/N} is particularly fast; notably, for very short options (under one month), $n_{max}=m_{max}=3$ is sufficient to get a precision of $10^{-2}$.

\begin{table}[ht]
 \caption{Short maturity cash-or-nothing prices, obtained via the series in formula~\ref{formula:asymmetric_C/N} truncated at $n=m=max$, positive and negative asymmetry; like in table~\ref{tab:symmetric_small_tau} (symmetric model), the series converge extremely fast. Parameters: $S=4200$, $K=4000$, $r=1\%$, $\sigma=0.2$, $\nu=0.85$.}
 \label{tab:asymmetric_small_tau}       
 \centering
 \begin{tabular}{lcccc}
  \hline\noalign{\smallskip}
 &  \multicolumn{4}{c}{Positive asymmetry [$\theta=+0.1$]}   \\
 & ${max}=3$  & ${max}=5$ & ${max}=10$ & ${max}=15$ \\
  \noalign{\smallskip}\hline\noalign{\smallskip}
 $\tau=\frac{1}{2}$ (6 months)  & 0.5370 & 0.5396 & 0.5398 & 0.5398 \\
 $\tau=\frac{1}{12}$ (1 month)  & 0.9401 & 0.9399 & 0.9399 & 0.9399 \\
 $\tau=\frac{1}{52}$ (1 week)   & 0.9873 & 0.9872 & 0.9872 & 0.9872 \\
 $\tau=\frac{1}{360}$ (1day)    & 0.9982 & 0.9982 & 0.9982 & 0.9982 \\
 \noalign{\smallskip}\hline
 \hline\noalign{\smallskip}
 &  \multicolumn{4}{c}{Negative asymmetry [$\theta=-0.1$]}   \\
 & ${max}=3$  & ${max}=5$ & ${max}=10$ & ${max}=15$ \\
 \noalign{\smallskip}\hline\noalign{\smallskip}
 $\tau=\frac{1}{2}$ (6 months)  & 0.7315 & 0.7290 & 0.7287 & 0.7287 \\
 $\tau=\frac{1}{12}$ (1 month)  & 0.9187 & 0.9184 & 0.9184 & 0.9184  \\
 $\tau=\frac{1}{52}$ (1 week)   & 0.9786 & 0.9786 & 0.9786 & 0.9786 \\
 $\tau=\frac{1}{360}$ (1day)    & 0.9968 & 0.9968 & 0.9968 & 0.9968 \\
 \noalign{\smallskip}\hline
\end{tabular}
\end{table}

\subsection{Gauss-Laguerre approximation for the European call}

Using a Gauss-Laguerre quadrature for the VG density function, \cite{Loregian12} obtain an approximation formula for the price of a European call, under the form of a linear convex combination of Black-Scholes prices (see Appendix \ref{app:LG} for details and notations). In the symmetric case ($\theta=0$), the formula reads
\begin{multline}\label{Loregian}
     C_{eur}(S,K,r,\tau,\sigma,\nu) \, \simeq \,
     S \, \sum\limits_{i=1}^n e^{ \overline{\omega}_{\sigma,\nu} \tau + \frac{\sigma^2}{2} \nu u_i } N(d_1(u_i)) p \left(u_i,\frac{\tau}{\nu} \right) 
    \, - \,
    Ke^{-r\tau} \, \sum\limits_{i=1}^n N(d_2(u_i)) \,  p \left(u_i,\frac{\tau}{\nu} \right) 
    .
\end{multline}

\begin{table}[ht]
 \caption{Comparisons between formula \eqref{formula:symmetric_European} truncated at $n=m=max$, and the Gauss-Laguerre quadrature \eqref{Loregian} truncated at $n=max$. Parameters: $K=4000$, $r=1\%$, $\tau=2$, $\sigma=0.2$, $\nu=0.85$.}
 \label{tab:Loregian}       
 \centering
 \begin{tabular}{lcccccc}
 \hline
 & \multicolumn{3}{c}{Formula \eqref{formula:symmetric_European}} & \multicolumn{3}{c}{Loregian \& al. \eqref{Loregian}} \\
  & $max=5$ & $max=10$ & $max=15$ & $max=5$ & $max=10$ & $max=15$  \\
  \hline
  $S=4500$ & 799.693  & 799.497 & 799.497 & 800.095 & 799.61 & 799.537 \\
  {\it relative error} & 0.02\% & 0.00\% & 0.00\% & 0.07\% & 0.01\% & 0.01\% \\
  \hline
  $ATM$    & 514.26 & 514.325 & 514.325 & 514.824 & 514.425 & 514.364 \\
  {\it relative error} & 0.01\% & 0.00\% & 0.00\% & 0.10\% & 0.02\% & 0.01\% \\
  \hline
  $S=3500$ & 234.803 & 232.197 & 232.197 & 232.768 &  232.267 & 232.214  \\
  {\it relative error}  & 1.12\% & 0.00\% & 0.00\% & 0.25\% & 0.03\% &  0.01\%  \\
 \hline
\end{tabular}
\end{table}

In table \ref{tab:Loregian}, the results obtained by truncations of the pricing formula \ref{formula:symmetric_European} and of the Gauss-Laguerre quadrature \eqref{Loregian} are compared and shown to provide excellent agreement. In terms of efficiency, the computation via our series formula \eqref{formula:symmetric_European} converges with a precision of $10^{-3}$ after $10$ iterations while the Gauss-Laguerre quadrature still presents a $0.01\%$ error after $15$ iterations; also note that the results are obtained instantaneously from formula \ref{formula:symmetric_European}, while letting $n\geq 10$ becomes time consuming when implementing the quadrature formula \eqref{Loregian}. Last, let us mention that if the shape parameter of the Gamma mixing density is smaller than $1$, then it is preferable to implement a generalized version of the Gauss-Laguerre quadrature, due to the presence of singularities around zero (see \cite{Rabinowitz67} and references therein).

\subsection{Monte Carlo pricing}

Monte Carlo simulations are well known and an abundant literature is available on both standard simulations and their variance reduction techniques, such as the antithetic variates method (see e.g. the classic reference \cite{Glasserman04}). In the particular case of the VG model, several algorithms for simulating the VG process are provided in \cite{Fu07}, based on time change Brownian motion, difference of Gamma processes or approximation of a compound Poisson process; in this paper we will rely on the VarianceGammaDistribution provided by Mathematica, whose density writes:
\begin{equation}
    \mathrm{VG} \left[ x, \lambda , \alpha , \beta, \mu   \right]
    \, := \,
    \frac{(\alpha^2 - \beta^2)^{\lambda}}{ \sqrt{\pi} \Gamma(\lambda) (2\alpha)^{\lambda - \frac{1}{2}} } \,
    e^{\beta (x-\mu)} \, | x-\mu |^{\lambda-\frac{1}{2}} \,
    K_{\lambda - \frac{1}{2}}  (\alpha | x - \mu| )
\end{equation}

\paragraph{Symmetric process} 
Define, for $n\in\mathbb{N}\backslash \{0\}$ and $i=1 \dots n$, 
\begin{equation}\label{MC_European}
     C_{eur}^{(i)} \, := \, e^{-r\tau} \,
    \left[ S e^{(r+\omega_{\sigma,\nu})\tau + Z^{(i)}} \, - \, K    \right]^+  
    \, , \hspace{0.5cm}
      C_{eur}^{(n)} \, := \, \frac{1}{n} \sum\limits_{i=1}^n C_{eur}^{(i)}
\end{equation}
where the $Z^{(i)}$ are independent and identically distributed according to the symmetric VG density
\begin{equation}
        Z^{(i)} \sim \mathrm{VG}\left[ x,\frac{\tau}{\nu},\frac{1}{\sigma}\sqrt{\frac{2}{\nu}},0,0\right ]
        .
\end{equation}
It is a consequence of the strong law of large numbers that $C_{eur}^{(n)} \longrightarrow \mathbb{E}_t^{\mathbb{Q}} \left[ e^{-r\tau} [S_T - K]^+\right]$ almost surely when $n\rightarrow\infty$; the 95\% confidence interval is $C_{eur}^{(n)} \pm 1.96 \, \sigma_P / \sqrt{n}$ where
\begin{equation}
    \sigma_P \, := \, \sqrt{\mathrm{var} \{ C_{eur}^{(i)} \}_{i=1 \dots n}  }
\end{equation}
and the confidence ratio is the length of the interval divided by $C_{eur}^{(n)}$. In table \ref{tab:MC} we compare the Monte Carlo prices for various numbers of paths, with truncations of the pricing formula \ref{formula:symmetric_European}; we observe that, as expected, formula \ref{formula:symmetric_European} shows a very good agreement with the Monte Carlo method, but converges far more rapidly and precisely. Moreover from $n \geq 10^4$, the high number of simulations strongly reduces the computational performance of the Monte Carlo method, and, in the OTM regions, confidence ratios are not as good as in the ITM region. For values of parameters like in table \ref{tab:MC}, the confidence ratio after $n=10^4$ simulations is of 0.9\% for the ITM cash-or-nothing call, but of nearly 5\% for the OTM European call. Variance in the OTM region can be reduced by implementing more sophisticated simulations, featuring for instance importance sampling methods \citep{Su00}.

\begin{table}[ht]
 \caption{Prices obtained for the European and the cash-or-nothing calls, via formulas \ref{formula:symmetric_European} (symmetric model), \ref{formula:asymmetric_C/N} (asymmetric model) and Monte Carlo simulations. We observe the accelerated convergence of our formulas when compared to the standard Monte Carlo simulations; after $n=10^4$ paths,  the confidence ratio is of $0.97\%$ for the ITM cash-or-nothing call but only of $2.7\%$ for the ITM European call, while formulas \ref{formula:symmetric_European} and \ref{formula:asymmetric_C/N} converge with of precision of $10^{-3}$ after 10 iterations.  Parameters: $K=4000$, $r=1\%$, $\tau=2$, $\sigma=0.2$, $\nu=0.85$.}
 \label{tab:MC}       
 \centering
 \begin{tabular}{lccc}
 \hline
 \multicolumn{4}{c}{{\bfseries Symmetric European call [$\theta = 0$}] }  \\
 \hline
 \multicolumn{4}{c}{Formula \ref{formula:symmetric_European}} \\
 & $max=5$ & $max=10$ & $max=15$   \\
 S=4500 & 799.693  & 799.497 & 799.497   \\
 ATM    & 514.26 & 514.325 & 514.325   \\
 S=3500 & 234.803 & 232.197 & 232.197  \\
 \hline
 \multicolumn{4}{c}{Monte Carlo simulation \eqref{MC_European}} \\
  & $n=10^2$ & $n=10^3$ & $n=10^4$  \\
 S=4500 & 757.729 & 784.85  & 801.309  \\
 ATM    & 618.979 & 512.819 & 514.772 \\
 S=3500 & 182.799 & 233.659 & 235.117 \\
 \hline
 \hline
 \multicolumn{4}{c}{{\bfseries Asymmetric cash-or-nothing call [$\theta = - 0.1$}] }  \\
 \hline
 \multicolumn{4}{c}{Formula \ref{formula:asymmetric_C/N}} \\
 & $max=5$ & $max=10$ & $max=15$   \\
 S=4500 & 0.662176 & 0.658984 & 0.658968  \\
 ATM    & 0.254457 & 0.251418 & 0.251402  \\
 S=3000 & 0.122406 & 0.123851 & 0.123843  \\
 \hline
 \multicolumn{4}{c}{Monte Carlo simulation \eqref{MC_C/N}} \\
  & $n=10^2$ & $n=10^3$ & $n=10^4$  \\
 S=4500 & 0.588119 & 0.668495 & 0.66036  \\
 ATM    & 0.303862 & 0.252891 & 0.253675  \\
 S=3000 & 0.107822 & 0.128406 & 0.120466  \\
 \hline
\end{tabular}
\end{table}

\paragraph{Asymmetric process.} The method remains the same for computing the cash-or-nothing call in the asymmetric VG model. The simulated payoff in this case reads:
\begin{equation}\label{MC_C/N}
     C_{c/n}^{(i)} \, := \, e^{-r\tau} \,
     \mathbbm{1}_{ \{  Z^{(i)} \, > \, -\log\frac{S}{K} - (r + \omega_{\sigma,\nu,\theta} ) \tau   \} }  
    \, , \hspace{0.5cm}
      C_{c/n}^{(n)} \, := \, \frac{1}{n} \sum\limits_{i=1}^n C_{c/n}^{(i)}
\end{equation}
where the $Z^{(i)}$ are independent and identically distributed according to the asymmetric VG density
\begin{equation}
        Z^{(i)} \sim \mathrm{VG}\left[ x,\frac{\tau}{\nu}, \frac{1}{\sigma^2} \sqrt{ \frac{2\sigma^2}{\nu} + \theta^2 } ,\frac{\theta}{\sigma^2},0\right ]
        .
\end{equation}
Comparisons between formula \ref{formula:asymmetric_C/N} and the Monte Carlo price \eqref{MC_C/N} are also provided in table \ref{tab:MC}.

\section{Concluding remarks}\label{sec:conclusion}

In this paper, we have provided the reader with several ready-to-use formulas for pricing path-independent options in the Variance Gamma model, with a particular focus on the symmetric model. These formulas, having the form of simple series expansions, can be easily implemented; the convergence is fast and can be made as precise as one wishes. Particularly remarkable is the short term behavior of the series: as it involves functions of powers of the moneyness, short maturities provide a very favorable situation and accelerate the convergence speed. This constitutes an interesting advantage over numerical evaluation of inverse Fourier integrals, which tend to oscillate and have a slower convergence when the time horizon is close to zero. Our approach also performs very well when compared to Monte Carlo methods: to get a confidence interval as precise as the results obtained by our formulas, one would need to simulate an extremely large number of paths, which would become time and resource consuming.

Future works should include the extension of the technology to other kinds of path-independent payoffs (notably on several assets, like correlation or spread options), and to exponential L\'evy models based on generalizations of the Variance Gamma process (e.g., multivariate processes) or on other L\'evy processes. Models based on Normal inverse Gaussian distributions as introduced in \cite{Barndorff-Nielsen98} are particularly interesting, because the density function admits a convenient Mellin-Barnes representation; like in the Variance Gamma case, this opens the way to a factorized option pricing formula and to its evaluation via the residue technique.

The extension of our approach to path-dependent options should also be investigated, notably for Asian options with continuous geometric payoffs. These options have notably been studied in \cite{Hubalek17} from the point of view of the Laplace transform in the framework of affine stochastic volatility models, resulting in a pricing formula in the Laplace space for the average price and the average strike, which may admit an equivalent formulation in the Mellin space. Moreover, the characteristic functions of the geometric average price have been derived explicitly in \cite{Fusai08} for several exponential L\'evy models (CGMY and their particular cases, Normal inverse Gaussian, $\alpha$-stable) for both fixed and floating strikes; these characteristic functions could also be suitable for obtaining a representation in the Mellin space for the probability densities, at least in the $\alpha$-stable case.


\section*{Acknowledgments}

The author thanks an anonymous reviewer for his careful reading of this paper and his valuable comments and suggestions.


\singlespacing

\doublespacing

\appendix

\section{Review of the Mellin transform}\label{app:Mellin}

The theory of the one-dimensional Mellin transform is explained in full detail in \cite{Flajolet95}, and table of Mellin transforms can be found in any monograph on integral transforms (see e.g. \cite{Bateman54}). An introduction to multidimensional complex analysis can be found in the classic textbook \cite{Griffiths78}, and applications to the specific case of Mellin-Barnes integrals are developed in \cite{Passare97,Zhdanov98}. 

\subsection{One-dimensional Mellin transform}

1. The Mellin transform of a locally continuous function $f$ defined on $\mathbb{R}^+$ is the function $f^*$ defined by
\begin{equation}\label{Mellin_def}
    f^*(s) \, := \, \int\limits_0^\infty \, f(x) \, x^{s-1} \, \ud x
    .
\end{equation}
The region of convergence $\{ \alpha < Re (s) < \beta \}$ into which the integral \eqref{Mellin_def} converges is often called the fundamental strip of the transform, and sometimes denoted $ < \alpha , \beta  > $.

\noindent 2. The Mellin transform of the exponential function is, by definition, the Euler Gamma function:
\begin{equation}
    \Gamma(s) \, = \, \int\limits_0^\infty \, e^{-x} \, x^{s-1} \, \ud x
\end{equation}
with strip of convergence $\{ Re(s) > 0 \}$. Outside of this strip, it can be analytically continued, except at every negative $s=-n$ integer where it admits the singular behavior
\begin{equation}\label{sing_Gamma}
    \Gamma(s) \, \underset{s\rightarrow -n}{\sim} \, \frac{(-1)^n}{n!}\frac{1}{s+n} \, ,
    \hspace*{0.5cm} n\in\mathbb{N}
    .
\end{equation}
In table \ref{tab:Mellin} we summarize the main Mellin transforms used in this paper, as well as their convergence strips.
\begin{table}[ht]
 \caption{Mellin pairs used throughout the paper. In all examples,  $Re(a) > 0$.}
 \label{tab:Mellin}       
 \centering
 \begin{tabular}{|c|c|c|}
 \hline
 $f(x)$  &  $f^*(s)$  &  Convergence strip  \\
 \hline   
 $e^{-a x}$       &  $a^{-s} \Gamma(s)$  &  $< 0 , \infty >$  \\
 $e^{-a x} - 1 $  &  $a^{-s} \Gamma(s)$  &  $< -1 , 0 >$  \\
 $K_{\nu}(ax)$    & $a^{-s}2^{s-2} \Gamma\left( \frac{s-\nu}{2} \right)\Gamma\left( \frac{s+\nu}{2} \right)$  &   $ < |Re(\nu)|, \infty > $ \\
 \hline
\end{tabular}
\end{table} 

\noindent 3. The inversion of the Mellin transform is performed via an integral along any vertical line in the strip of convergence:
\begin{equation}\label{inversion}
    f(x) \, = \, \int\limits_{c-i\infty}^{c+i\infty} \, f^*(s) \, x^{-s} \, \frac{\ud s}{2i\pi} \hspace*{1cm} c\in ( \alpha, \beta )
\end{equation}
and notably for the exponential function one gets the so-called \textit{Cahen-Mellin integral}:
\begin{equation}\label{Cahen}
    e^{-x} \, = \, \int\limits_{c-i\infty}^{c+i\infty} \, \Gamma(s) \, x^{-s} \, \frac{\ud s}{2i\pi}, \hspace*{0.5cm} c>0
    .
\end{equation}

\noindent 4. When $f^*(s)$ is a ratio of products of Gamma functions of linear arguments:
\begin{equation}
    f^*(s) \, = \, \frac{\Gamma(a_1 s + b_1) \dots \Gamma(a_m s + b_m)}{\Gamma(c_1 s + d_1) \dots \Gamma(c_l s + d_l)}
\end{equation}
then one speaks of a \textit{Mellin-Barnes integral}, whose \textit{characteristic quantity} is defined to be
\begin{equation}
    \Delta \, = \, \sum\limits_{k=1}^m \, a_k \, - \, \sum\limits_{j=1}^l \, c_j
    .
\end{equation}
$\Delta$ governs the behavior of $f^*(s)$ when $|s|\rightarrow \infty$ and thus the possibility of computing \eqref{inversion} by summing the residues of the analytic continuation of $f^*(s)$ right or left of the convergence strip:
\begin{equation}
    \left\{
    \begin{aligned}
        & \Delta < 0 \hspace*{1cm} f(x) \, = \, -\sum\limits_{Re(s) > \beta} \, \res \left[f^*(s)x^{-s}\right],  \\
        & \Delta > 0 \hspace*{1cm} f(x) \, = \, \sum\limits_{Re(s) < \alpha} \, \res \left[f^*(s)x^{-s}\right]
        .
    \end{aligned}
    \right.
\end{equation}
For instance, in the case of the Cahen-Mellin integral one has $\Delta = 1$ and therefore:
\begin{equation}
e^{-x} \, = \, \sum\limits_{Re(s)<0} \res \left[ \Gamma(s) \, x^{-s} \right] \,  =  \, \sum\limits_{n=0}^{\infty} \, \frac{(-1)^n}{n!}x^n
\end{equation}
as expected from the usual Taylor series of the exponential function.

\subsection{Multidimensional Mellin transform}

1. Let $n\in\mathbb{N}$, and let us denote the vectors in $\mathbb{C}^n$ by $\underline{s}$. Let the $\underline{a}_k$, $\underline{\tilde{a}}_j\in\mathbb{C}^n$, let the $b_k$, $\tilde{b}_j$ be complex numbers and let "." represent the euclidean scalar product. We speak of a Mellin-Barnes integral in $\mathbb{C}^n$ when one deals with an integral of the type
\begin{equation}
     \int\limits_{\underline{c}+i\mathbb{R}^n} \, \omega
\end{equation}
where $\omega$ is a complex differential n-form that reads
\begin{equation}
     \omega \, = \, \frac{\Gamma(\underline{a}_1.\underline{s} + b_1) \dots \Gamma(\underline{a}_m.\underline{s} + b_m)}{\Gamma(\underline{\tilde{a}}_1.\underline{s} + \tilde{b}_1) \dots \Gamma(\underline{\tilde{a}}_l.\underline{s} + \tilde{b}_l)} \, x_1^{-s_1} \, \dots \, x_n^{-s_n} \, \frac{\ud s_1 \dots \ud s_n}{(2i\pi)^n}  \hspace*{1cm} \, x_1, \dots x_n \in\mathbb{R}
     .
\end{equation}
The singular sets induced by the singularities of the Gamma functions
\begin{equation}
     D_k \, := \, \{ \underline{s}\in\mathbb{C}^n \, , \, \underline{a}_k.\underline{} + b_k = -n_k \, , \, n_k \in\mathbb{N}   \} \hspace*{1cm}  k=1 \dots m
\end{equation}
are called the \textit{divisors} of $\omega$. The \textit{characteristic vector} of $\omega$ is defined to be
\begin{equation}
     \Delta \, = \, \sum\limits_{k=1}^m \underline{a}_k \, - \, \sum\limits_{j=1}^l \underline{\tilde{a}}_j
\end{equation}
and the \textit{admissible half-plane}:
\begin{equation}
     \Pi_\Delta \, := \, \{ \underline{s}\in\mathbb{C}^n \, , \, Re( \Delta . \underline{s} ) \, < \, Re( \Delta . \underline{c})  \}
     .
\end{equation}

\noindent 2. Let the $\rho_k$ be real numbers, the $h_k:\mathbb{C}\rightarrow\mathbb{C}$ be linear applications and $\Pi_k$ be a subset of $\mathbb{C}^n$ of the type
\begin{equation}\label{Pik}
    \Pi_k \, := \, \{ \underline{s}\in\mathbb{C}^n, \, Re(h_k(s_k)) \, < \, \rho_k \}
    .
\end{equation}
A \textit{cone} in $\mathbb{C}^n$ is a Cartesian product
\begin{equation}
    \Pi \, = \, \Pi_1 \times \dots \times \Pi_n
    ,
\end{equation}
where the $\Pi_k$ are of the type \eqref{Pik}. Its \textit{faces} $\varphi_k$ are
\begin{equation}
    \varphi_k \, := \, \partial \Pi_k \hspace*{1cm} k=1,\dots , n
\end{equation}
and its \textit{distinguished boundary}, or \textit{vertex} is
\begin{equation}
    \partial_0 \, \Pi \, := \, \varphi_1 \, \cap \, \dots \cap \, \varphi_n
    .
\end{equation}

\noindent 3. Let us now assume $n=2$ for simplicity. We group the divisors $D=\cup_{k=1}^m \, D_k$ of the complex differential form $\omega$ into two sub-families
\begin{equation}
    D_1 \, := \, \cup_{k=1}^{m_0} \, D_k \, , \hspace{0.5cm} D_2 \, := \, \cup_{k=m_0+1}^{m} \, D_k \, , \hspace*{0.5cm}  D \, = \, D_1\cup D_2
    .
\end{equation}
We say that a cone $\Pi\subset\mathbb{C}^2$ is \textit{compatible} with the family of divisors $D$ if:
\begin{enumerate}
\item[-] \, Its distinguished boundary is $\underline{c}$,
\item[-] \, Every divisor $D_1$ and $D_2$ intersects at most one of his faces:
\begin{equation}
    D_1 \, \cap \, \varphi_1 \, = \, D_2 \, \cap \, \varphi_2 \, = \, \emptyset 
    .
\end{equation}
\end{enumerate}

\noindent 4. Residue theorem for multidimensional Mellin-Barnes integral \citep{Passare97,Zhdanov98}: If $\Delta \neq 0$ and if $\Pi\subset\Pi_\Delta$ is a compatible cone located into the admissible half-plane, then
\begin{equation}
     \int\limits_{\underline{c}+i\mathbb{R}^2} \, \omega \, = \, \sum\limits_{\Pi} \mathrm{Res} \, \omega
\end{equation}
and the series is absolutely convergent. The residues are to be understood as the "natural" generalization of the Cauchy residue, that is:
\begin{equation}
    \left. \res \, \left[ f(s_1,s_2) \, \frac{\ud s_1}{2i\pi t_1^{n_1}} \wedge \frac{\ud s_2}{2i\pi s_1^{n_2}}  \right] \, = \, \frac{1}{(n_1-1)!(n_2-1)!}\frac{\partial ^{n_1+n_2-2} f}{\partial s_1^{n_1-1}   \partial s_2^{n_2-1} }\right\vert_{\substack{s_1=0\\s_2=0}}
    .
\end{equation}

\section{Gauss-Laguerre quadrature}\label{app:LG}
The Gauss-Laguerre quadrature (see details e.g. in \cite{Abramowitz72}) is an extension of the Gaussian quadrature rule, allowing to approximate integrals involving an exponential term on the positive real axis:
\begin{equation}
    \int\limits_0^{\infty} \, f(u) e^{-u} \, \ud u \, \simeq \sum\limits_{i=1}^n \, w(u_i) f(u_i)
\end{equation}
where:
\begin{itemize}
    \item[-] the $u_i$, $i=1 \dots n$ are the zeros of the n$^{th}$ Laguerre polynomial $L_n(u)$;
    \item[-] the weights $w(u_i)$ are defined by:
    \begin{equation}
        w(u_i) \, := \, \frac{u_i}{(n+1)^2 L_{n+1}^2(u_i)}
        .
    \end{equation}
\end{itemize}

\begin{table}[ht]
 \caption{Zeros $u_i$ and weights $w(u_i)$ of the Gauss-Laguerre quadrature up to $n=6$.}
 \label{tab:GL_coeffs}       
 \centering
 \begin{tabular}{|cll|}
 \hline
 $n$ & $u_i$   & $w(u_i)$  \\
 \hline
 1 & 1 & 1 \\
 \hline
 2 & 0.585786 & 0.853557 \\
   & 3.41421 & 0.146447 \\
 \hline
 3 & 0.415775 & 0.711088  \\
   & 2.29428 & 0.278518 \\
   & 6.28995 & 0.010389 \\
 \hline
 4 & 0.322548 & 0.603149 \\
   & 1.74576 & 0.35413 \\
   & 4.53662 & 0.038879 \\
   & 9.39507 & 0.000539295 \\
 \hline
 5 & 0.26356 & 0.521763 \\
   & 1.4134 & 0.398676 \\
   & 3.59643 & 0.0759415 \\
   & 7.08581& 0.00361176 \\
   & 12.6408 & 0.00002337 \\
 \hline
  6  & 0.222847 & 0.458954 \\
  & 1.18893 & 0.41701 \\
  & 2.99274 & 0.113372 \\
  & 5.77514 & 0.010399 \\
  & 9.83747 & 0.000261 \\
  & 15.9829 & 8.9852 $\times 10^{-7}$ \\
 \hline
\end{tabular}
\end{table}

The zeros $u_i$ and their associated weights $w(u_i)$ are computed in table \ref{tab:GL_coeffs} up to $n=6$. Starting from the integral formula \eqref{VG_density_1}, \cite{Loregian12} use a Gauss-Laguerre quadrature to re-write the VG density function; if
\begin{equation}
    \varphi(x,\alpha,\beta) \, := \, \frac{1}{\beta\sqrt{2\pi}} \, e^{- \frac{(x-\alpha)^2}{2\beta^2} }
\end{equation}
is the Gaussian density and if we introduce
\begin{equation}
    p \left(u_i, t \right) \, := \, \frac{ w(u_i)u_i^{ t -1} }{ \sum_{i=1}^n w(u_i)u_i^{ t -1} }
    ,
\end{equation}
then the VG density \eqref{VG_density_1} can be approximated by
\begin{equation}\label{VG_density_LG_approx}
    \overline{f}_{\sigma,\nu,\theta} (x,t) \, := \, \sum\limits_{i=1}^n \, \varphi(x , \theta \nu u_i , \sigma\sqrt{\nu u_i} ) \, p \left(u_i,\frac{t}{\nu} \right)
    .
\end{equation}
The martingale adjustment can be approximated by:
\begin{equation}\label{omega_LG_approx}
    \overline{\omega}_{\sigma,\nu,\theta} \, := \, -\frac{1}{t} \, \log \, \sum\limits_{i=1}^n  \, p \left(u_i,\frac{t}{\nu} \right) \, e^{ \theta\nu u_i + \frac{\sigma^2}{2} \nu u_i  }
\end{equation}
and the European call price can be written down as:
\begin{multline}\label{European_LG_approx}
   \mathbb{E}_t^{\mathbb{Q}} \, [  e^{-r\tau} [S_T - K]^+ ] \,
    \simeq
    \\
    S \, \sum\limits_{i=1}^n e^{ \theta\nu u_i + \overline{\omega}_{\sigma,\nu,\theta} \tau + \frac{\sigma^2}{2} \nu u_i } N(d_1(u_i)) p \left(u_i,\frac{\tau}{\nu} \right) 
    \, - \,
    Ke^{-r\tau} \, \sum\limits_{i=1}^n N(d_2(u_i)) \,  p \left(u_i,\frac{\tau}{\nu} \right) 
    ,
\end{multline}
where $x\rightarrow N(x)$ is the Normal cumulative distribution function, and where
\begin{equation}
    \left\{
    \begin{aligned}
    & d_1(u_i) :=   \frac{ \log\frac{S}{K} + (r + \overline{\omega}_{\sigma,\nu,\theta}) \tau + ( \sigma^2 + \theta ) \nu u_i }{ \sigma\sqrt{\nu u_i} }    \\
    & d_2(u_i) := d_1(u_i) - \sigma\sqrt{\nu u_i}
    .
    \end{aligned}
    \right.
\end{equation}

\end{document}